\documentclass[conference]{IEEEtran}
\IEEEoverridecommandlockouts
\usepackage{cite}
\usepackage{amsmath}
\usepackage{amssymb,amsfonts}
\usepackage{amsthm}
\usepackage{algorithmic}
\usepackage{graphicx}
\usepackage{textcomp}
\usepackage{xcolor}
\usepackage{braket}
\usepackage[hidelinks]{hyperref}
\usepackage[T1]{fontenc}
\usepackage{cleveref}

\crefformat{equation}{(#2)#1(#3)}
\crefrangeformat{equation}{(#3#1#4)--(#5#2#6)}

\def\BibTeX{{\rm B\kern-.05em{\sc i\kern-.025em b}\kern-.08em
    T\kern-.1667em\lower.7ex\hbox{E}\kern-.125emX}}

\newcommand{\nb}{\bar{n}_\mathrm{B}}
\newcommand{\im}{j}
\newcommand{\todo}[1]{\textbf{\textcolor{purple}{TODO: #1}}}

\newtheorem{theorem}{Theorem}
\newtheorem{lemma}{Lemma}

\DeclareMathOperator{\tr}{tr}

\allowdisplaybreaks

\renewcommand{\thesubsubsection}{\arabic{subsubsection}}

\makeatletter 
\newcommand{\linebreakand}{%
  \end{@IEEEauthorhalign}
  \hfill\mbox{}\par
  \mbox{}\hfill\begin{@IEEEauthorhalign}
}
\makeatother 

\renewcommand{\thesubsubsection}{\arabic{subsubsection}}

\begin{document}

\title{Covert Quantum Communication Over Optical Channels
\thanks{This work was supported by the National Science Foundation under Grants No. CCF-2006679 and EEC-1941583.}}

\author{\IEEEauthorblockN{Evan J. D. Anderson}
\IEEEauthorblockA{\textit{Wyant College of Optical Sciences} \\
\textit{University of Arizona}\\
Tucson, AZ, USA \\
ejdanderson@arizona.edu}
\and
\IEEEauthorblockN{Christopher K. Eyre}
\IEEEauthorblockA{\textit{Dept. of Mathematics} \\
\textit{ Brigham Young University}\\
Provo, UT, USA}\\
\and 
\IEEEauthorblockN{Isabel M. Dailey}
\IEEEauthorblockA{\textit{Dept. of Electrical and Computer Engineering} \\
\textit{College of Engineering} \\
\textit{University of Arizona}\\
Tucson, AZ, USA}\\
\and

\linebreakand 

\IEEEauthorblockN{Filip Rozpędek}
\IEEEauthorblockA{\textit{College of Information \& Computer Sciences} \\
\textit{University of Massachusetts Amherst}\\
Amherst, MA, USA}\\
\and
\IEEEauthorblockN{Boulat A. Bash}
\IEEEauthorblockA{\textit{Dept. of Electrical and Computer Engineering} \\
\textit{College of Engineering} \\
\textit{Wyant College of Optical Sciences} \\
\textit{University of Arizona}\\
Tucson, AZ, USA}
}


\maketitle

\begin{abstract}
We explore covert communication of qubits over the lossy thermal-noise bosonic channel, which is a quantum-mechanical model of many practical channels, including optical. Covert communication ensures that an adversary is unable to detect the presence of transmissions, which are concealed in channel noise. We show a \emph{square root law} (SRL) for quantum covert communication similar to that for classical: $\propto\sqrt{n}$ qubits can be transmitted covertly and reliably over $n$ uses of an optical channel. Our achievability proof uses photonic dual-rail qubit encoding, which has been proposed for long-range repeater-based quantum communication and entanglement distribution. Our converse employs prior covert signal power limit results and adapts well-known methods to upper bound quantum capacity of optical channels. Finally, we believe that the gap between our lower and upper bounds for the number of reliable covert qubits can be mitigated by improving the quantum error correction codes and quantum channel capacity bounds.
\end{abstract}


\section{Introduction}
Covert, or low probability of detection/intercept (LPD/LPI) communication renders adversaries unaware of the presence of transmission between two or more parties. The last decade saw much exploration of the fundamental limits of covert communication over classical channels, with \cite{bash15covertcommmag,bash13squarerootjsac,bloch15covert, wang15covert} leading to many follow-on works. However, the physics which underpins these channels is quantum.  This motivated recent work on covert classical-quantum channels \cite{bash15covertbosoniccomm, bullock20discretemod, gagatsos20codingcovcomm, azadeh16quantumcovert-isitarxiv, bullockCovertCommunicationClassicalQuantum2023}. For all these channels, covert communication is fundamentally governed by the \textit{square root law} (SRL) which limits communication that is both covert and reliable to $\propto\sqrt{n}$ bits over $n$ channel uses. In this paper, we extend these results to quantum covert communication over lossy thermal-noise bosonic channels. Such channels model optical fiber, and free space communication in the optical, microwave, and radio-frequency regimes. Here, we study the achievability of quantum covert communication over such a channel utilizing dual-rail photonic qubits, as well as the converse using well-known upper bounds on the quantum capacity of a lossy thermal-noise bosonic channel. 

Dual-rail qubits apply to many quantum information processing tasks. They are used in cluster-state generation \cite{thomasEfficientGenerationEntangled2022}, which, e.g., can enable one-way quantum computing. Additionally, the dual-rail encoding is convenient for entanglement distribution in quantum networks. Indeed, it was used to demonstrate a loophole-free Bell inequality violation \cite{hensenLoopholefreeBellInequality2015}, and to entangle trapped-ion qubits spatially separated by \text{230 m} \cite{krutyanskiyEntanglementTrappedIonQubits2023}. Furthermore, dual-rail qubits have been proposed to transmit quantum information over long-range repeater-based quantum networks \cite{takedaDeterministicQuantumTeleportation2013, guhaRateLoss2015, dhara2023entangling, azuma2023repeatersurvey}. The dual-rail encoding is also commonly used in quantum key distribution (QKD) \cite{scarani09rmpQKD, Honjo08qkd}.

Covert quantum communication has been previously explored in the context of QKD \cite{arrazolaCovertQuantumCommunication2016, tahmasbi19covertqkd,tahmasbi20bosoniccovertqkd-jsait,tahmasbi20covertqkd}. Here, however, we take a direct approach and use dual-rail encoding to address quantum covert communication over the lossy thermal-noise bosonic channel. 
Specifically, we adapt the analysis from covert classical-quantum channels \cite{bash15covertbosoniccomm, bullock20discretemod, gagatsos20codingcovcomm, azadeh16quantumcovert-isitarxiv, bullockCovertCommunicationClassicalQuantum2023}. 
We believe that this approach extends naturally to other quantum encodings and channels. Analogous to the $\sqrt{n}$ scaling for bits in the classical-classical and classical-quantum channels, we find achievability of the SRL, where one can transmit reliably at least $\propto \sqrt{n}$ covert qubits with the dual-rail encoding over $n$ uses of a lossy thermal-noise bosonic channel. 

In the converse, we use the upper bound on the number of photons per mode that is covertly transmissible over a lossy thermal-noise bosonic channel \cite{bullock20discretemod}. Adapting the upper bound on the energy-constrained quantum capacity of lossy thermal-noise bosonic channel \cite{sharmaBoundingEnergyconstrainedQuantum2018} shows that at most $\propto \sqrt{n}$ covert qubits can be reliably sent over $n$ uses of this channel, matching the achievable lower bound scaling. However, the gap between the achievability and the converse remains open. 

The rest of this paper is organized as follows: in Section~\ref{sec:preliminaries} we provide the mathematical preliminaries as well as the system and channel models. In Section~\ref{sec:covertcommunication} we describe the mathematical formalism underpinning covert communication and provide our results. Finally, we wrap up in Section~\ref{sec:discussion} with a discussion of our results and future research.

\section{Preliminaries}
\label{sec:preliminaries}

\subsection{Dual-rail Qubits}
\label{sec:dual-rail}

The dual-rail qubit is a well-known encoding of qubits into single photons in linear-optical quantum computing \cite{knillSchemeEfficientQuantum2001} and quantum communication  \cite{takedaDeterministicQuantumTeleportation2013}. With dual-rail encoding, a qubit is represented by the presence of a single photon in one of two optical modes. The logical states are physically represented using two-mode Fock (photon number) states: $|0\rangle_L = |01\rangle$ and $|1\rangle_L = |10\rangle$. The general logical qubit state $|\psi\rangle$ is then a superposition of the two states: 
\begin{align}
|\psi\rangle = \alpha|0\rangle_L + \beta|1\rangle_L = \alpha|01\rangle + \beta|10\rangle \label{eq:dual-rail-qubit}
\end{align}
and $\hat{\rho}_{\alpha,\beta} = |\psi\rangle\langle\psi|$ is the state's density operator with $\alpha, \beta \in \mathbb{C}$ normalized such that $|\alpha|^2 + |\beta|^2 = 1$.

We call our fundamental transmission unit a \emph{round}.  We transmit one qubit per round, occupying two optical modes.

\subsection{System and Channel Model}
\begin{figure*}[htb]
\centering
\includegraphics[width=\textwidth]{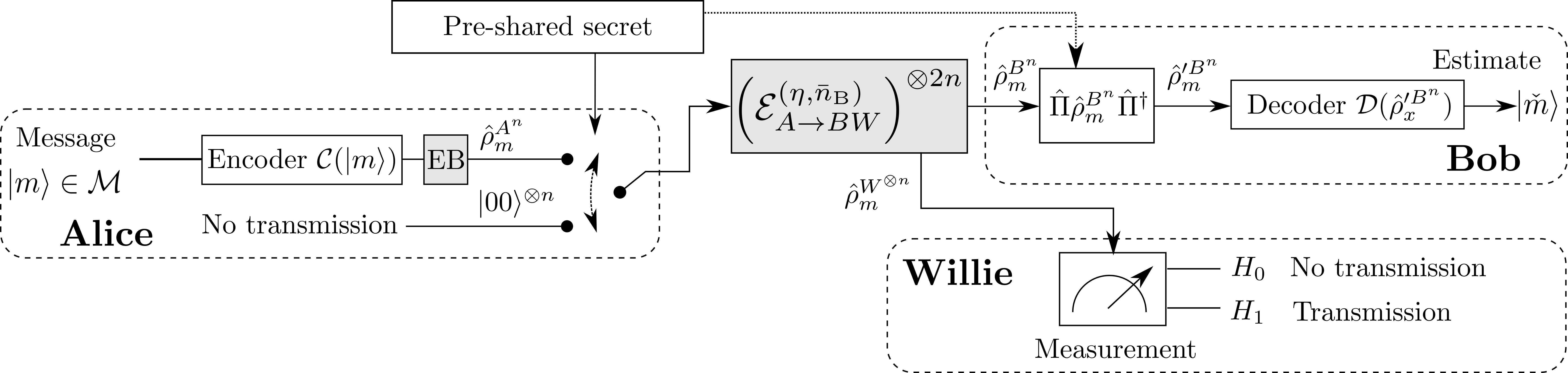}
\caption{Covert quantum communication over the lossy thermal-noise bosonic channel $\mathcal{E}^{(\eta,\bar{n}_{\rm B})}_{A\to BW}$. 
Alice and Bob employ a random coding scheme from \cite[Secs.~23.3 and 24.4]{wilde16quantumit2ed}. Transmission occurs with probability $q$ in each of the $n$ two-mode rounds. The rounds chosen for transmission constitute a pre-shared classical secret unknown to Willie.  Alice chooses a quantum message $|m\rangle \in \mathcal{M}$.
Alice employs additional noise or loss before transmitting to ensure that her channel to Willie is entanglement breaking. This is represented by the block labeled ``EB.'' For each of the chosen rounds, Bob performs a probabilistic projection represented by operator $\hat{\Pi}$ to the dual-rail basis in \eqref{eq:dual-rail-qubit}. On failure, he replaces the state with a completely mixed state. He then decodes to obtain an estimate $|\check{m}\rangle$ of the original message. Willie tests between hypotheses $H_0$ and $H_1$ to determine if Alice is transmitting. }
\label{fig:system}
\end{figure*}

Alice employs blocks of $n$ two-mode rounds to encode each covert quantum message $|m\rangle$ using dual-rail qubits described in Section \ref{sec:dual-rail} and vacuum states $|00\rangle\langle00|$.
Thus, she employs a total of $2n$ optical modes. Utilizing the pre-shared classical secret as described in Section \ref{sec:covertcommunication}, she either sends a dual-rail qubit or vacuum $|00\rangle\langle00|$ through the lossy thermal-noise bosonic channel, as detailed in Fig.~\ref{fig:system}. 
The channel acts on each optical mode independently. Bob attempts to decode his received state utilizing the shared secret to obtain an estimate $|\check{m}\rangle$ of the message, while the adversary warden Willie tries to detect Alice's transmission.

Consider a channel, $\mathcal{E}_{A\to BW}^{(\eta,\bar{n}_\textrm{B})}$, in Fig.~\ref{fig:system} that is described by a beamsplitter with transmittance $\eta \in [0,1]$, two input modes (Alice and the environment), and two output modes (Bob and Willie). These modes are labeled by their modal annihilation operators $\hat{a}, \hat{e}, \hat{b}$, and $\hat{w}$ respectively.  Their input-output modal relationships are:
\begin{align}
     \hat{b} = \sqrt{\eta}\hat{a}+\sqrt{1-\eta}\hat{e} \ \text{ and }  \ 
     \hat{w} = \sqrt{1-\eta}\hat{a}-\sqrt{\eta}\hat{e} \label{eq:bsmodal-willie}.
\end{align}
For the \emph{lossy thermal-noise bosonic channel}, the input state of mode $\hat{e}$ is $\hat{\rho}_{\bar{n}_{\textrm{B}}}$, a zero-mean thermal state with mean photon number $\bar{n}_{\textrm{B}}$ expressed by the following sum over the diagonal elements in the Fock (photon number) basis $|k\rangle$:
\begin{align}
\hat{\rho}_{\bar{n}_{\textrm{B}}} &\equiv  \sum_{k=0}^{\infty} t_k(\bar{n}_{\textrm{B}}) |k\rangle\langle k|, \text{~where~} 
    t_k(\bar{n}_{\textrm{B}}) = \frac{\bar{n}_{\textrm{B}}^k}{(1+\bar{n}_{\textrm{B}})^{k+1}}. \label{eq:thermCoefficient}
\end{align}
A \emph{pure-loss} bosonic channel has $\bar{n}_{\textrm{B}}=0$ and mixes the input with vacuum rather than thermal noise.
Finally, we note that the lossy thermal-noise bosonic channel belongs to a well-studied class of bosonic Gaussian channels \cite{weedbrook12gaussianQIrmp}.

\subsection{Entanglement-Breaking Channel}
An entanglement-breaking channel breaks entanglement between input quantum states at the channel's output. Entanglement is broken in a lossy thermal-noise bosonic channel if $\bar{n} > \kappa$ where $\bar{n} > 0$ is the mean photon number of the thermal noise added by the channel and $\kappa$ is the fraction of the input photon number at the output \cite{holevoEntanglementbreakingChannelsInfinite2008}. In the following analysis, we assume that the Alice-to-Willie channel is naturally entanglement breaking, corresponding to $\eta\bar{n}_{\textrm{B}} > 1-\eta$,  as is typical in optical communication systems. It is unknown whether the channel to the adversary must be entanglement breaking for covert quantum communication, and is a subject of ongoing investigation. However, if the physical channel does not break entanglement, Alice may introduce additional loss or noise after encoding to ensure entanglement is broken, per the following two lemmas:

\begin{lemma} \label{lemma:ent-break-trans}
    Entanglement is broken in the Alice-to-Willie channel by passing Alice's signal through a pure-loss channel with transmittance $\tau < \frac{\eta}{1-\eta}\nb$ prior to transmission, meeting the entanglement-breaking condition  $\eta\bar{n}_{\textrm{B}} > \tau(1-\eta)$.
\end{lemma}
\begin{lemma} \label{lemma:ent-break-noise}
    Entanglement is broken in the Alice-to-Willie channel by passing Alice's signal through a quantum-limited amplifier with gain coefficient $G_{\rm eb}=2(1-\eta)/(2(1-\eta)-\eta\bar{n}_{\textrm{B}}^\prime)$ and $\bar{n}_{\textrm{B}}^\prime > \frac{\eta}{1-\eta}-\nb$, followed by a pure loss channel with transmittance $\tau = 1/G_{\rm eb}$. This meets the entanglement-breaking condition $\eta(\bar{n}_{\textrm{B}}+\bar{n}_{\textrm{B}}^\prime) > (1-\eta)$.
\end{lemma}

Proofs of both lemmas are in \ref{sec:induced-entanglement-breaking}. Enforcing the entanglement-breaking condition via either lemma reduces the number of qubits transmitted reliably and covertly without affecting the achievable SRL scaling. In Lemma~\ref{lemma:ent-break-trans} attenuation of Alice's signal breaks entanglement, while in Lemma~\ref{lemma:ent-break-noise} additional noise is used. Intuitively, Lemma~\ref{lemma:ent-break-trans} holds because a lossy thermal-noise bosonic channel decomposes into a pure-loss channel followed by a quantum-limited amplifier \cite{garciapatron2012majorization, caruso2006weakdegradability, rosatiNarrowBoundsQuantum2018}. We combine the pure-loss component with Alice's additional pure-loss channel. 
The resulting transmittance is the fraction of Alice's input photon number delivered to Willie, while the amplifier gain determines the thermal noise added by the channel. Proof of Lemma~\ref{lemma:ent-break-noise} is not as intuitive.

\section{Covert Communication}\label{sec:covertcommunication}
\subsection{Covertness Analysis}
\label{sec:covertnessanalysis}
Denote $\hat{\rho}_0^{W^n}$ and $\hat{\rho}_1^{W^n}$ as the respective states Willie observes when Alice is quiet or transmitting.
Since two-mode vacuum $|00\rangle\langle00|$  is input when Alice is quiet, $\hat{\rho}_0^{W^n}=\left(\hat{\rho}_0^W\right)^{\otimes n}$, where $\hat{\rho}_0^W = \hat{\rho}_{\eta\bar{n}_{\textrm{B}}} \otimes \hat{\rho}_{\eta\bar{n}_{\textrm{B}}}$ is a two-mode thermal product state \cite{weedbrook12gaussianQIrmp}.
Willie desires to determine if Alice and Bob are communicating; ergo, in $n$ rounds (uses of the two-mode bosonic channel), he tries to distinguish between $\hat{\rho}_1^{W^n}$ and $\left(\hat{\rho}_0^W\right)^{\otimes n}$. The null and alternate hypotheses $H_0$ and $H_1$ correspond to Alice being quiet and transmitting, respectively. Willie collects all the photons that do not reach Bob, as shown in Fig.~\ref{fig:system}. 

Willie can make two types of errors: a false alarm, where he decides that Alice is transmitting when she is not (choosing $H_1$ when $H_0$ is true), and a missed detection, where he decides that Alice is not transmitting when she is (choosing $H_0$ when $H_1$ is true). As is customary in the literature, we assume equal prior probabilities for the hypotheses $P(H_0) = P(H_1)=\frac{1}{2}$, though this is not a requirement. Thus, Willie's probability of error is: $P_e = \frac{P(H_0|H_1) + P(H_1|H_0)}{2}$.

Willie guessing randomly yields an ineffective detector with $P_e = \frac{1}{2}$. Hence, Alice's goal is to transmit so that $P_e$ is as close to $\frac{1}{2}$ as possible. Formally, we call any system covert if, for large enough $n$ and $\delta > 0$, $P_e \ge \frac{1}{2} - \delta$. With access to a quantum-optimal detector, Willie can achieve minimal $P_e = \frac{1}{2} - \frac{1}{4} \left\|\hat{\rho}_1^{W^{n}}-\left(\hat{\rho}_0^W\right)^{\otimes n}\right\|_1$ \cite[Sec.~9.1.4]{wilde16quantumit2ed}, where $\|\hat{A}\|_1\equiv\tr\left[\sqrt{\hat{A}^\dagger\hat{A}}\right]$ is the trace norm of $\hat{A}$ \cite[Def. 9.1.1]{wilde16quantumit2ed}. This implies our system is covert if $\frac{1}{4} \left\|\hat{\rho}_1^{W^n}-\left(\hat{\rho}_0^W\right)^{\otimes n}\right\|_1 \le \delta$. 

The trace distance is often mathematically unwieldy. Quantum relative entropy (QRE), $D\left(\hat{\rho}\middle\|\hat{\sigma}\right) = \operatorname{tr}[\hat{\rho}\log\hat{\rho} - \hat{\rho}\log\hat{\sigma}]$, is commonly employed in covertness analysis \cite{bash15covertbosoniccomm, bullock20discretemod, gagatsos20codingcovcomm}, since it is additive over product states and upper bounds the trace distance via the quantum Pinsker's inequality \cite[Th. 11.9.1]{wilde16quantumit2ed}:
\begin{align}
\frac{1}{4}\left\|\hat{\rho}_1^{W^{n}}-\left(\hat{\rho}_0^W\right)^{\otimes n}\right\|_1 \leq \sqrt{\frac{1}{8} D\left(\hat{\rho}_1^{W^{n}}\middle\|\left(\hat{\rho}_0^W\right)^{\otimes n}\right)}.\label{eq:pinskers}
\end{align}
We employ the \emph{right-hand side} (r.h.s.) of \eqref{eq:pinskers} rather than the left-hand side (l.h.s.) as our covertness criterion, as is common in both the classical \cite{bloch15covert, wang15covert} and quantum \cite{bullock20discretemod, gagatsos20codingcovcomm, azadeh16quantumcovert-isitarxiv, bullockCovertCommunicationClassicalQuantum2023} analyses.
That is, we upper bound the r.h.s.~of \eqref{eq:pinskers} by $\delta$.

Alice and Bob ensure the covert communication by randomly selecting the transmission rounds they use via $n$ flips of a biased coin, with the probability of heads $q$ to be determined later. If the $i^{\text{th}}$ flip is heads, then the $i^{\text{th}}$ round is selected. 
%
The chosen dual-rail systems constitute the classical pre-shared secret in Fig.~\ref{fig:system}.
A quantum error correction code (QECC) is used, with transmissions taking place only in the selected rounds.
This procedure, channel parameters, value of $q$, the QECC, and the time of transmission are known to Willie.

As Willie's channel from Alice is entanglement breaking, the state he observes is a classical superposition of product states given by
\begin{align}
     \hat{\rho}_1^{W^n} = \sum_{ \ket{m}\in\mathcal{M}} p\left(\ket{m}\right)\bigotimes_{i=1}^{n}\hat{\rho}_{1,i}^W(m), \label{eq:willie-n-expansion}
\end{align}
where $\mathcal{M}$ is the set of messages, 
and $\sum_{\ket{m}\in\mathcal{M}} p\left(\ket{m}\right) = 1$. Furthermore, 
\begin{align}
    \hat{\rho}_{1,i}^W (m)&= (1-q)\hat{\rho}_0^{W} + q \hat{\rho}_{i,m}^W,\label{eq:rho1i}
\end{align}
where  $\hat{\rho}_{i,m}^W$ is Willie's two-mode state when Alice transmits a (possibly mixed) state with density matrix
\begin{align}
    \hat{\rho}_{i,m}^A = \begin{pmatrix} |\alpha(i,m)|^2 & \gamma(i,m) \\ \gamma^\ast(i,m)& |\beta(i,m)|^2 \end{pmatrix},
\end{align}
in the logical basis with arbitrary coefficients $\alpha(i,m)$, $\beta(i,m)$, and $\gamma(i,m)$ depending on the QECC.
Then, using \eqref{eq:willie-n-expansion} yields the following upper bound on the QRE in \eqref{eq:pinskers}:
\begin{IEEEeqnarray}{rCl}
\IEEEeqnarraymulticol{3}{l}{D\left( \hat{\rho}_1^{W^n} \middle\| \left(\hat{\rho}_0^W\right)^{\otimes n} \right)}\IEEEnonumber \\&\le& \sum_{\ket{m}\in\mathcal{M}} p\left(\ket{m}\right) D\left( \bigotimes_{i=1}^{n}\hat{\rho}_{1,i}^W(m)\middle\|\left(\hat{\rho}_0^W\right)^{\otimes n}\right)  \label{eq:QRE-convexity}\\
    &=& \sum_{\ket{m}\in\mathcal{M}} p\left(\ket{m}\right) \sum_{i=1}^{n} D\left( \hat{\rho}_{1,i}^W(m)\middle\|\left(\hat{\rho}_0^W\right)^{\otimes n}\right)  \label{eq:QRE-additivity}\\
&\le& \sum_{\ket{m}\in\mathcal{M}} p\left(\ket{m}\right)  q^2\sum_{i=1}^{n}   D_{\chi^2}\left(\hat{\rho}_{i,m}^{W} \middle\| \hat{\rho}_0^W\right), \label{eq:chi2-bound}
\end{IEEEeqnarray}
where \eqref{eq:QRE-convexity} and \eqref{eq:QRE-additivity} follow from the convexity \cite[Corollary 11.9.2]{wilde16quantumit2ed} and additivity \cite[Ex.~11.8.7]{wilde16quantumit2ed} properties of the QRE. Lastly, \eqref{eq:chi2-bound} is by \cite[Lemma 1]{bullockCovertCommunicationClassicalQuantum2023},
where the quantum $\chi^2$-divergence \cite{temme2010chi2} between two states $\hat{\rho}$ and $\hat{\sigma}$ is:
\begin{align}
    D_{\chi^2}\left(\hat{\rho}\middle\|\hat{\sigma}\right) = \operatorname{tr}\left[(\hat{\rho}-\hat{\sigma})^2 \hat{\sigma}^{-1}\right] = \operatorname{tr}[\hat{\rho}^2\hat{\sigma}^{-1}] - 1, \label{eq:chisquare}
\end{align}
with the second equality due to the cyclic property of the trace and the fact that the trace of a quantum state is unity.
The following yields a bound on $D_{\chi^2}\left(\hat{\rho}_{i,m}^{W} \middle\| \hat{\rho}_0^W\right)$:

\begin{lemma}\label{lemma:ccov} When Alice transmits an arbitrary quantum state with the following density operator in the logical basis
\begin{align}
    \hat{\rho}^A = \begin{pmatrix} |\alpha|^2 & \gamma \\ \gamma^\ast& |\beta|^2 \end{pmatrix}\label{eq:alice-arb-input}
\end{align}
that is encoded in dual-rail qubit basis, the corresponding output at Willie $\hat{\rho}^W$ satisfies:
\begin{align}
    D_{\chi^2}\left(\hat{\rho}^W\middle\|\left(\hat{\rho}_0^W\right)^{\otimes 2}\right) &\leq \frac{(1-\eta)^2}{\eta \bar{n}_{\textrm{B}} (1+\eta \bar{n}_{\textrm{B}})} \label{eq:chisquared-qubit}.
\end{align}
\end{lemma}
Here we provide a proof sketch, with the full proof deferred to \ref{sec:proof-of-lemma3}.


\begin{IEEEproof}[Proof (sketch)] 
Our first challenge is to determine $\hat{\rho}^W$. First, we find its density-operator representation. The anti-normally ordered characteristic function completely defines a quantum state $\hat{\rho}$, and, for a two-mode state, is given by:
\begin{align}
\chi_A^{\hat{\rho}}(\zeta_1,\zeta_2) &= \tr\left[\hat{\rho}e^{-\zeta_1^* \hat{a}_1}e^{\zeta_1 \hat{a}_1^{\dagger}}e^{-\zeta_2^* \hat{a}_2}e^{\zeta_2 \hat{a}_2^{\dagger}}\right],\label{eq:charfunc}
\end{align}
where $\zeta_i \in \mathbb{C}$ and $\hat{a}_i, \hat{a}_i^\dagger$ are modal annihilation and creation operators for $i=1,2$ \cite{weedbrook12gaussianQIrmp}. Using the expression for $\hat{w}$ in \eqref{eq:bsmodal-willie},
\begin{align}
\chi_A^{\hat{\rho}^W}(\zeta_1,\zeta_2) &=  \chi_A^{\hat{\rho}^{A}}\left(\sqrt{1-\eta}\zeta_1,\sqrt{1-\eta}\zeta_2\right) \notag \\ 
&\phantom{=}\times\chi_A^{\hat{\rho}^E}\left(\sqrt{\eta}\zeta_1,\sqrt{\eta}\zeta_2\right),
\end{align}
where $\chi_A^{\hat{\rho}^{A}}(\cdot)$ and $\chi_A^{\hat{\rho}^{E}}(\cdot)$ are the characteristic functions for Alice's input state \eqref{eq:alice-arb-input} and the thermal state. The expression for $\chi_A^{\hat{\rho}^{E}}(\cdot)$ is well known \cite[Sec.~7.4.3.2]{orszag16quantumotpics}. 
$\chi_A^{\hat{\rho}^{A}}(\cdot)$ is derived by using \eqref{eq:alice-arb-input} in \eqref{eq:charfunc} and expanding the exponentials:
\begin{align}
    \chi_A^{\hat{\rho}^W}(\zeta_1,\zeta_2) &=  e^{-(1+\eta \bar{n}_{\textrm{B}})(|\zeta_1|^2+|\zeta_2|^2)}\left[1- (1-\eta)\left(|\alpha|^2||\zeta_2|^2 \right . \right . \notag \\
    &\phantom{=} \left . \left .\vphantom{|\beta|^2}   
     + |\beta|^2|\zeta_1|^2+\gamma\zeta_1\zeta_2^* +\gamma^*\zeta_1^*\zeta_2\right)  \right ].\label{eq:charfuncwillie-simplified}
\end{align}

A quantum state $\hat{\rho}^W$ and its characteristic function $\chi_A^{\hat{\rho}^W}(\cdot)$ are related via the operator Fourier transform \cite{weedbrook12gaussianQIrmp}:
\begin{align}
    \hat{\rho}^W &= \iint \frac{d^2\zeta_1}{\pi}\frac{d^2\zeta_2}{\pi}\chi_A^{\hat{\rho}^W} e^{\zeta_2\hat{w}_2^\dagger} e^{\zeta_1\hat{w}_1^\dagger} e^{-\zeta_1^*\hat{w}_1} e^{-\zeta_2^*\hat{w}_2}, \label{eq:fourier-transform}
\end{align}
where the integrals are over the complex planes for $\zeta_1$ and $\zeta_2$. This allows $\hat{\rho}^W$ to be expressed in the Fock basis with the elements $p_{f,g,f^\prime,g^\prime}= \langle fg | \hat{\rho}^W | f^\prime g^\prime \rangle$ for $f,g,f^\prime,g^\prime \in \mathbb{N}_0$.

The integrals in \eqref{eq:fourier-transform} are evaluated in polar coordinates with expansions of the exponentials that include annihilation and creation operators. Details are in \ref{sec:proof-of-lemma3}. Due to the orthogonality of Fock states for the $|\alpha|^2$ and $|\beta|^2$ contributions, the only non-zero terms occur when $f=f^\prime$ and $g=g^\prime$, defining the main diagonal of the density operator. The $\gamma$-contribution terms are non-zero when $f'=f+1$ and $g'=g-1$ from an off-by-one exponential in integration over the corresponding polar coordinates. The $\gamma^*$-contributions follow similarly for $f'=f-1$ and $g'=g+1$. The Fourier transform in \eqref{eq:fourier-transform}  yields:
\begin{align}
    \hat{\rho}^W &= \sum_{g=0}^\infty \sum_{f=0}^\infty  \left(|\alpha|^2 W_1(f, g) + |\beta|^2 W_1(g, f)\right) |fg\rangle\langle fg| \notag \\
    &\hphantom{= \sum_{g=0}^\infty \sum_{f=0}^\infty} + \gamma W_2(g,f)|fg\rangle\langle f+1,g-1| \notag \\
    &\hphantom{= \sum_{g=0}^\infty \sum_{f=0}^\infty} + \gamma^* W_2(f,g) |fg\rangle\langle f-1,g+1|  \label{eq:SumTrace}
\end{align}
where
\begin{align}
W_1(f, g) &=  \left( \frac{(\eta \bar{n}_{\textrm{B}})^g}{(1+\eta \bar{n}_{\textrm{B}})^{g+1}} - \frac{(1-\eta)(\eta \bar{n}_{\textrm{B}}-g)(\eta \bar{n}_{\textrm{B}})^{g-1}}{(1+\eta \bar{n}_{\textrm{B}})^{g+2}} \right) \notag \\ 
     &\times \frac{(\eta \bar{n}_{\textrm{B}})^f}{(1+\eta \bar{n}_{\textrm{B}})^{f+1}}, \notag \\ 
     W_2(f,g)&=\frac{(1-\eta)(\eta \bar{n}_{\textrm{B}})^{g+f-1}}{(1+\eta \bar{n}_{\textrm{B}})^{g+f+3}}\sqrt{f(g+1)}.
\end{align}
Thus, $\hat{\rho}^W$ is a tri-diagonal operator as it is defined by $|fg\rangle\langle fg|$, $|fg\rangle\langle f-1,g+1|$, and $|fg\rangle\langle f+1,g-1|$. We obtain $(\hat{\rho}^W)^2$ by assigning each diagonal to operators $\hat{A}$, $\hat{B}$, $\hat{C}$ and computing $(\hat{\rho}^W)^2 = \left(\hat{A}+\hat{B}+\hat{C}\right)^2$.

Now, the density operator for the two-mode thermal state $\left(\hat{\rho}_0^{W}\right)^{\otimes 2}$ received by Willie when Alice is silent is \cite{weedbrook12gaussianQIrmp}:
\begin{align}
\left(\hat{\rho}_0^{W}\right)^{\otimes 2} = \sum_{f=0}^{\infty}\sum_{g=0}^{\infty} t_g(\eta\bar{n}_{\textrm{B}}) t_f(\eta\bar{n}_{\textrm{B}}) |fg\rangle\langle fg|, \label{eq:innocentWillie}
\end{align}
where $t_f(\eta\bar{n}_{\textrm{B}})$ and $t_g(\eta\bar{n}_{\textrm{B}})$ are defined in \eqref{eq:thermCoefficient}. Since \eqref{eq:innocentWillie} is a diagonal operator, its inverse is also diagonal. Calculating $D_{\chi^2}\left(\hat{\rho}^W\middle\|\left(\hat{\rho}_0^{W}\right)^{\otimes 2}\right) = \tr\left[(\hat{\rho}^W)^2 \left((\hat{\rho}_0^{W})^{-1}\right)^{\otimes 2}\right] - 1$ reduces to multiplying the diagonal elements of $(\hat{\rho}^W)^2$ and $\left((\hat{\rho}_0^{W})^{-1}\right)^{\otimes 2}$, and summing the results. This yields:
\begin{align}
&\tr\left[(\hat{\rho}^W)^2 \left((\hat{\rho}_0^{W})^{-1}\right)^{\otimes 2}\right] = \left[\left((1-\eta)^2+\eta\nb(1+\eta\nb)\right) \right .\notag \\
&\phantom{=}\times (|\alpha|^4+|\beta|^4) + 2\left(|\alpha|^2|\beta|^2\eta\nb(1+\eta\nb) \right . \notag \\
&\phantom{=}\left.\left.+ (1-\eta)^2|\gamma|^4\right)\right] /(\eta\nb(1+\eta\nb)) \label{eq:SummedTraceTogether}.
\end{align}
A pure-state logical qubit input with $|\gamma| = |\alpha\beta|$ maximizes \eqref{eq:SummedTraceTogether}, yielding the lemma.
\end{IEEEproof}


Lemma~\ref{lemma:ccov} upper-bounds $D_{\chi^2}\left(\hat{\rho}_{i,m}^{W} \middle\| \hat{\rho}_0^W\right)$ independently of $\alpha$, $\beta$, and $\gamma$.
Combining it with \eqref{eq:pinskers} and \eqref{eq:chi2-bound} yields:
\begin{align}
\sqrt{\frac{1}{8} D\left(\hat{\rho}_1^{W^{n}}||\left(\hat{\rho}_0^W\right)^{\otimes n}\right)} & \leq \frac{q(1-\eta)\sqrt{n}}{2\sqrt{2\eta \bar{n}_{\textrm{B}} (1+\eta \bar{n}_{\textrm{B}})}}.\label{eq:traceinequality2}
\end{align}
Therefore, the covertness requirement is maintained if $q \le \frac{2c_{\text{cov}}\delta}{\sqrt{n}}$, where the covertness constant is:
\begin{align}
    c_{\text{cov}} &= \frac{\sqrt{2\eta \bar{n}_{\textrm{B}} (1+\eta \bar{n}_{\textrm{B}})}}{(1-\eta)}. \label{eq:ccov}
\end{align}
Note that we employ the same constant $c_{\text{cov}}$ as in \cite[Eq.~(2)]{bullock20discretemod}.

\subsection{Reliability Analysis and Achievability} \label{sec:reliability}

Let $M(n)$ be the number of qubits transmitted covertly in $n$ channel uses. Denote by $[x]^+=\max(x,0)$. The following theorem characterizes the lower bound on $E[M(n)]$, where the expectation is over the biased random coin flips used to select the dual-rail systems in Section \ref{sec:covertnessanalysis}:

\begin{theorem}[Achievability] \label{thm:lower-bound}
$E[M(n)]\geq{2\sqrt{n}c_{\mathrm{cov}}R\delta}$ qubits can be transmitted reliably and covertly over $n$ uses of the lossy thermal-noise bosonic channel, where $c_\mathrm{cov}$ is in \eqref{eq:ccov}, and $\delta$ is the covertness constraint. $R\geq\left[1-H(\vec{p})\right]^+$ is the constant achievable rate of reliable qubit transmission per round, where $\vec{p}=\left[1-\frac{3p}{4},\frac{p}{4},\frac{p}{4},\frac{p}{4}\right]$, $p=1-\frac{\eta}{(1+(1-\eta)\nb)^4}$, and $H(\vec{p})=-\sum_{p_i\in\vec{p}} p_i\log(p_i)$ is the Shannon entropy. 
\end{theorem}

\begin{IEEEproof}
Alice and Bob pre-share a secret, determining the rounds to be used for transmission. Alice employs a random code from \cite[Secs.~23.3 and 24.4]{wilde16quantumit2ed} to encode the message. The expected number of rounds selected is $qn={2c_{\mathrm{cov}}\delta}{\sqrt{n}}$, per Section \ref{sec:covertnessanalysis}.

Bob projects the two-mode systems in each of the selected rounds into the subspace spanned by the dual-rail basis states in \eqref{eq:dual-rail-qubit}. The probability of projection failure is $p_\mathrm{fail} = 1-\langle 01|\hat{\rho}^B|01\rangle+\langle 10|\hat{\rho}^B|10\rangle = 1-\frac{2\nb(1+\nb)(1-\eta)^2+\eta}{(1+(1-\eta)\nb)^4}$ where $\hat{\rho}^B$ is described by \eqref{eq:SumTrace} with $\eta$ swapped for $1-\eta$ and vice versa, and arbitrary $\alpha$, $\beta$, and $\gamma$. When the projection is unsuccessful, Bob replaces the state with the maximally mixed state $\frac{\hat{\pi}}{2}$. This mimics a depolarizing channel $\hat{\rho}^B \to \hat{\rho}^{B}_\mathrm{proj} = (1-p_\mathrm{fail}) \hat{\rho}^B + p_\mathrm{fail}\frac{\hat{\pi}}{2}$ parameterized by $p_\mathrm{fail}$.

Furthermore, this projection allows one to treat the lossy thermal-noise channel as a depolarizing channel acting on $\hat{\rho}^{B}_\mathrm{proj}$ and parameterized by $p^\prime = \frac{2(1-\eta)^2\bar{n}_{\textrm{B}}(1+\bar{n}_{\textrm{B}})}{\eta+2(1-\eta)^2\bar{n}_{\textrm{B}}(1+\bar{n}_{\textrm{B}})}$ \cite[Appendix B]{kish2023comparison}. Then Bob's state, $\hat{\rho}^{\prime B}$, prior to decoding is given by  $\hat{\rho}^{B}_\mathrm{proj} \to \hat{\rho}^{\prime B} = (1-p^\prime)\hat{\rho}^{B}_\mathrm{proj} + p^\prime \frac{\hat{\pi}}{2}$, where
\begin{align}
\hat{\rho}^{\prime B} &= (1-p^\prime)\left((1-p_\mathrm{fail}) \hat{\rho}^B + p_\mathrm{fail}\frac{\hat{\pi}}{2}\right) + p^\prime \frac{\hat{\pi}}{2} \\
&= (1-p^\prime)(1-p_\mathrm{fail})\hat{\rho}^B+(p^\prime+(1-p^\prime)p_\mathrm{fail})\frac{\hat{\pi}}{2}\\
&= (1-p)\hat{\rho}^B+p\frac{\hat{\pi}}{2}, \label{eq:bob-depolarized-state}
\end{align}
with $p=p^\prime+(1-p^\prime)p_\mathrm{fail}=1-\frac{\eta}{(1+(1-\eta)\nb)^4}$ in \eqref{eq:bob-depolarized-state}. Thus, in each round, Bob's state is equivalent to Alice's state transmitted through a depolarizing channel parameterized by $p$. This channel is a Pauli channel parameterized by $\vec{p}=\left[1-\frac{3p}{4},\frac{p}{4},\frac{p}{4},\frac{p}{4}\right]$ \cite[Ex. 4.7.4]{wilde16quantumit2ed}. We complete the proof using the hashing bound \cite[Sec.~24.6.3]{wilde16quantumit2ed}.
\end{IEEEproof}

The following remarks are in order:
\subsubsection{Entanglement-breaking condition}
If the Alice-to-Willie channel is not naturally entanglement breaking, Alice may use Lemmas~\ref{lemma:ent-break-trans} or \ref{lemma:ent-break-noise} to break it.
Using Lemma~\ref{lemma:ent-break-trans} replaces $(1-\eta)$ with $\tau(1-\eta)$ in the expressions for $c_\mathrm{cov}$ and $\eta$ by $1-\tau(1-\eta)$ in $R$.  Using Lemma~\ref{lemma:ent-break-noise} replaces $\bar{n}_\mathrm{B}$ with $\bar{n}_\mathrm{B}+\bar{n}^\prime_\mathrm{B}$ in $c_\mathrm{cov}$ and $\bar{n}_\mathrm{B}+\bar{n}^{\prime\prime}_\mathrm{B}$ in $R$ where $\bar{n}^{\prime\prime}_\mathrm{B}= \frac{2(1-1/G)\eta}{1-\eta}$.

\subsubsection{Use of auxiliary covert classical channel} \label{rem:classical_comms}
Suppose that Alice and Bob have a covert full-duplex classical communication link. 
Let Alice prepare Bell states, sending one qubit of each state to Bob using the dual-rail basis on the rounds selected for transmission.
Bob's projection of his received state to the smaller dual-rail subspace in \eqref{eq:dual-rail-qubit} is probabilistic, with outcomes known to Bob. 
Bob communicates to Alice the indices of successful rounds over the covert classical full-duplex link, allowing hashing-based entanglement distillation \cite{bennet96qec} on these rounds. Alice can then teleport qubits to Bob using this distributed entanglement and the covert classical link, achieving
$R^\prime= (1-p_\mathrm{fail})\left(1-H(\vec{p^\prime})\right)$ where $\vec{p^\prime}=\left[1-\frac{3p^\prime}{4},\frac{p^\prime}{4},\frac{p^\prime}{4},\frac{p^\prime}{4}\right]$ and $p^\prime$ is the same as that in the proof of Theorem~\ref{thm:lower-bound}. The resulting expected number of reliably-transmissible covert qubits is plotted in Fig.~\ref{fig:bounds}.
The average number of classical bits that need to be covertly exchanged is $\propto \sqrt{n}$, making this scheme feasible under certain channel conditions.
However, we defer the characterization of covert classical communication link requirements to future work.


\subsubsection{Potential improvement of the bound on \texorpdfstring{$R$}{R} without classical communication} \label{rem:lower-bound}
When Bob fails to project a state to the basis defined by $\eqref{eq:dual-rail-qubit}$, he replaces it with a maximally mixed state. Hence, instead of an erasure error, it is treated as a random Pauli error and Bob throws away useful information that may aid in decoding. Indeed, it is known that any stabilizer code can correct up to twice
as many erasure errors as Pauli errors \cite[Sec.~III.A]{grasslCodesQuantumErasure1997}. Although the use of random codes and their analytical achievable rate $R$ has only been established for Pauli errors, the rates for large codes correcting erasure errors (such as the tree code \cite{varnavaLossToleranceOneWay2006}) can be computed numerically.

Finally, here, we restrict encoding to the finite-dimensional subspace spanned by dual-rail basis in \eqref{eq:dual-rail-qubit}. Bosonic codes may improve the rate $R$, as they take advantage of the entire infinite-dimensional space. However, the covertness requirement may diminish this advantage. For example, ideal Gottesman-Kitaev-Preskill (GKP) states require infinite energy, conflicting with covertness. Nevertheless, the trade-off should be studied.

\subsection{Converse}
\label{sec:converse}
The following provides an upper bound on $M(n)$:
\begin{theorem}[Converse] \label{thm:upper-bound}
$M(n) \le 2nC$ where 
\begin{align}
C &= \left[g\left(\frac{(G+1)\bar{n}_\mathrm{S}+\bar{G}}{2} \right)-g\left(\frac{\bar{G}(1+\bar{n}_\mathrm{S})}{2}\right)\right]^+, \label{eq:upperbound}
\end{align}
with $g(x)\equiv (1+x)\log(1+x)-x\log(x)$, $G=\frac{\eta}{\eta-(1-\eta)\nb/2}$, $\bar{G}=G-1$, and the mean photon number of Alice's input state for $c_{\rm cov}$ defined in \eqref{eq:ccov} is constrained by $\bar{n}_\textrm{S} \leq \frac{2c_{\mathrm{cov}}\delta}{\sqrt{n}}$.
\end{theorem}



\begin{IEEEproof}[Proof] In Theorem \ref{thm:lower-bound}, we use the total of $2n$ modes of lossy thermal-noise bosonic channel, since we transmit dual-rail qubits. Hence, we employ the standard arguments from \cite[Sec.~24.5]{wilde16quantumit2ed} to obtain $M(n) \le C_{2n}^\prime$, where $C_{2n}^\prime$ is the channel quantum capacity (regularized coherent information) over these $2n$ uses.
Any lossy thermal-noise channel with transmittance $\eta$ and mean thermal photon number $\bar{n}_\textrm{B}$ can be decomposed into a quantum-limited amplifier with gain coefficient $G=\frac{\eta}{\eta-(1-\eta)\nb/2}$ followed by a pure-loss channel with transmittance $\eta^\prime = \eta/G$ \cite{garciapatron2012majorization, caruso2006weakdegradability}.
Discarding one of these channels and applying the data-processing inequality upper bounds $C_{2n}^\prime\leq C_{2n}$, where $C_{2n}$ is the quantum capacity of the remaining channel over $2n$ uses. Since both pure-loss and pure-input quantum-limited amplifier channels are degradable \cite{caruso2006degradability}, their coherent information is additive, and $C_{2n}=2nC$, where $C$ is the single-channel-use quantum capacity.
Usually, the amplifier channel is discarded (see, e.g.,
 \cite[Th.~16]{sharmaBoundingEnergyconstrainedQuantum2018}), since the resulting bound is tighter unless $\bar{n}_{\rm S}\to0$.
 Here, due to  constraint on $\bar{n}_{\rm S}$, we obtain a tighter bound by discarding the pure-loss\footnote{Discarding the amplifier instead yields a poor bound that is $\propto\sqrt{n}\log(n)$.} channel: $C$ in \eqref{eq:upperbound} is an upper bound on the quantum capacity of the amplifier channel with gain $G$ assisted by the arbitrary local
operations and classical communication (LOCC) \cite[Eqs. (115), (172)]{davis18constrainedcap}\cite{goodenough2016quantumrepeat}.

The energy-constrained capacity is required, as bounding the QRE on the r.h.s.~of \eqref{eq:pinskers} limits Alice's input state mean photon number $\bar{n}_{\rm S}$ per \cite[Thm. 1]{bullock20discretemod}. While this result\footnote{There is a typo in the short paragraph between Criterion 2 and Eq.~(4) in \cite{bullock20discretemod}: $\delta=\sqrt{\delta_{\rm QRE}}$ should be $\sqrt{8}\delta=\sqrt{\delta_{\rm QRE}}$. We apply this correction in deriving the constraint on $\bar{n}_{\rm S}$ in the statement of Theorem \ref{thm:upper-bound}.} is applied to the classical-quantum capacity in \cite{bullock20discretemod}, the theorem is general for any quantum state. 
\end{IEEEproof}

\begin{figure}[htb]
\centering
\includegraphics[width=8.5cm]{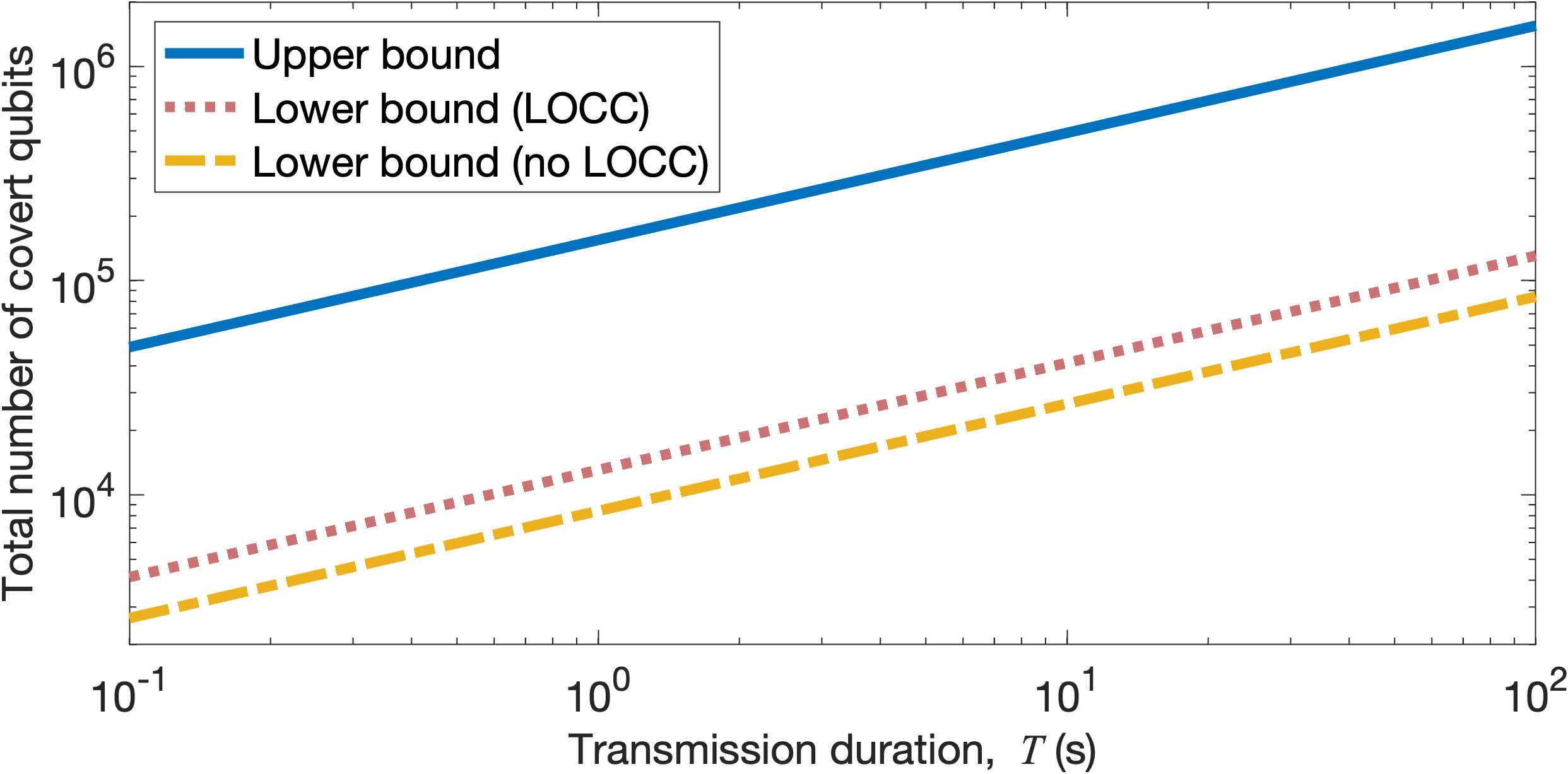}
\caption{Log-log plot for upper (solid blue line) and lower (dot-dashed red line) bounds given by Theorem~\ref{thm:upper-bound} and Theorem~\ref{thm:lower-bound} respectively for the total number of covert bits reliably transferred vs time in seconds. The orange dotted line is the lower bound discussed in remark~\ref{rem:classical_comms} of Section~\ref{sec:reliability} enabled by a two-way classical communication link. Here, transmittance $\eta=0.9$, the mean thermal photon number $\bar{n}_{\textrm{B}} = 0.12$, and the covertness parameter $\delta = 0.05$.}
\label{fig:bounds}
\end{figure}

Taylor series expansion of \eqref{eq:upperbound} around $\bar{n}_{\rm S}=0$ yields the SRL scaling of the upper bound in Theorem~\ref{thm:upper-bound}, matching that of the lower bound in Theorem~\ref{thm:lower-bound}.
However, a multiplicative gap exists between these bounds, as shown in Fig.~\ref{fig:bounds},
where we set the channel transmittance to $\eta = 0.9$, and $\bar{n}_{\textrm{B}}=0.12$. These parameters ensure that the Alice-to-Willie channel is entanglement breaking. 
Furthermore, we employ a modulation frequency of 100 MHz or $10^{8}$ modes/second with a covertness criterion of $\delta=0.05$. We suspect the looseness of the upper bound is due to the data processing argument, insofar as we disregard the pure-loss channel in the decomposition of the original channel. Indeed, a tighter bound may exist. However, deriving such a bound remains a difficult open problem as it requires analysis of the regularized coherent information of the channel \cite[Th.~24.3.1]{wilde16quantumit2ed} \cite{khatriInformationtheoreticAspectsGeneralized2020}.

\section{Conclusion and Discussion}

\label{sec:discussion}
We develop an achievable lower bound on the expected number $E[M(n)]$ of qubits that are covertly and reliably transmissible using dual-rail qubit encoding over the lossy thermal-noise bosonic channel (we defer  removing the expectation to future work).
We also provide a converse. 
Although we specifically address quantum communication rather than QKD, we expect our work to provide insight into the open questions in covert QKD \cite{arrazolaCovertQuantumCommunication2016,tahmasbi19covertqkd,tahmasbi20bosoniccovertqkd-jsait,tahmasbi20covertqkd}.

While both the upper and lower bounds in the converse and the achievability scale $\propto \sqrt{n}$, the gap between them is fairly large. This motivates improvement of the QECC capabilities and the upper bounds of the quantum capacity of the lossy thermal-noise bosonic channel, as noted in Sections \ref{sec:reliability} and \ref{sec:converse}.  
Furthermore, in our achievability analysis, we require that the Alice-to-Willie channel is entanglement breaking to simplify the mathematics. While Lemmas~\ref{lemma:ent-break-trans} and \ref{lemma:ent-break-noise} may be employed to ensure this condition, for atmospheric models such as MODTRAN \cite{berk06MODTRAN}, $R=0$ in Theorem~\ref{thm:lower-bound}. However, using classical covert channel can yield $R>0$. Further investigation of the necessity of entanglement-breaking condition is needed. 
Other covertness criteria, including bounding the trace norm on the l.h.s.~of \eqref{eq:pinskers} directly (as done in \cite{wangCharacterizationCovertCapacity2022} for classical-quantum channels), also need to be studied.



Practical aspects of achieving covert quantum communication have to be considered.  Our covertness scheme requires a substantial number of classical pre-shared secret bits; resolvability techniques from \cite{bloch15covert,bullockCovertCommunicationClassicalQuantum2023} should be adapted to reduce this burden. Additionally, here we assume that quantum states are generated for transmission on demand.  However, quantum processes are inherently random.  This stochasticity needs to be included in the calculation of channel use selection probability $q$.  Moreover, although randomness in state generation is generally considered an undesired characteristic of quantum information processing, here it might be exploited. For example, the quasi-probabilistic nature of heralded entangled photonic Einstein-Podolsky-Rosen (EPR) pair generation \cite{Chen} could be used to select a random subset of channel uses.

Finally, the framework provided in this manuscript can be applied to other qubit encodings and quantum channels. Encodings such as single-rail and GKP are enticing candidates due to their prevalence in the literature and error-correcting properties of the latter. In fact, it has been shown \cite{noh2018capacitybounds} that GKP qubits allow reliable quantum communication rates that differ only by a constant factor from the known upper bound on the quantum capacity of the pure-loss bosonic channel and perform well in the lossy thermal-noise bosonic channel setting. However, their energy requirements negatively impact covertness.
Indeed, practical codes enabling a physical realization of covert quantum communication should be investigated.
 
\section*{Acknowledgement}
The authors benefited from discussions with Christos Gagatsos, Brian Smith, Ryan Camacho, Michael Bullock, Narayanan Rengaswamy, Kenneth Goodenough, and Saikat Guha.

\vfill \clearpage

\appendices


\renewcommand{\thesection}{Appendix \Alph{section}}

\section{Induced Entanglement Breaking Channel} \label{sec:induced-entanglement-breaking}

\subsection{Proof of Lemma \ref{lemma:ent-break-trans}} \label{proof-lemma-ent-trans}
\begin{figure}[htb]
\centering
\includegraphics[width=5cm]{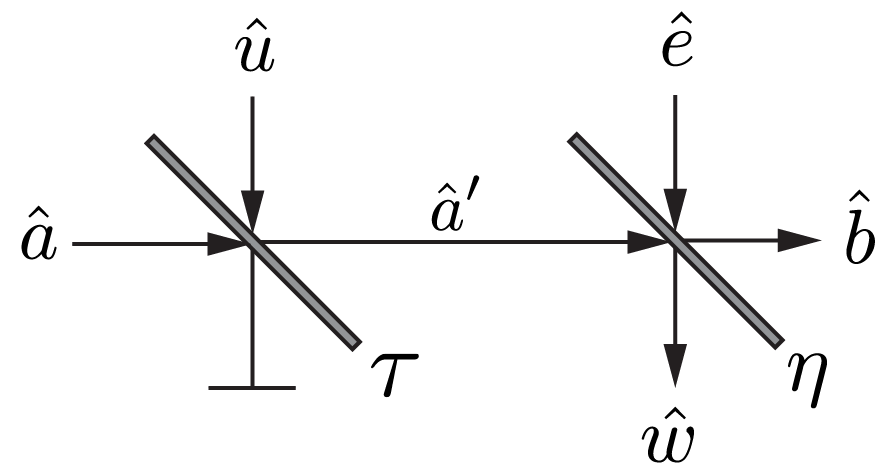}
\caption{A system Alice can employ to induce an entanglement breaking channel. Annihilation operators $\hat{a}, \hat{u}$ and $\hat{e}$ label Alice’s and the environment’s input modes, while $\hat{w}$ and $\hat{b}$ label Bob’s and warden Willie’s output modes. $\hat{a}^\prime$ labels an intermediate mode that is the output of the first beamsplitter with transmissivity $\tau$ and input into the beamsplitter with transmissivity $\eta$. Alice has control over the beamsplitter transmissivity $\tau$ and the input state at the mode $\hat{u}$. The beamsplitter with transmissivity $\eta$ and input state at $\hat{e}$ represent the original system channel in Fig.~\ref{fig:system} and are fixed.}
\label{fig:ent-break-loss}
\end{figure}

Alice's system is in Fig.~\ref{fig:ent-break-loss}. She inputs a state in mode $\hat{a}$ of a beamsplitter with transmissivity $\tau$. Its output passes through a second beamsplitter with reflectance $1-\eta$. Alice controls $\tau$ and its second input mode, $\hat{u}$, which is in a vacuum state. $\eta$ is fixed, as is the thermal state $\hat{\rho}_{\nb}$ of the input mode $\hat{e}$. 

The complementary bosonic channel characterized by $\eta$ and $\nb$ decomposes into a pure-loss bosonic channel followed by a quantum-limited amplifier \cite{garciapatron2012majorization, caruso2006weakdegradability, rosatiNarrowBoundsQuantum2018}
with gain coefficient $G=\eta\frac{\nb}{2}+1$.
It is followed by a pure-loss channel with reflectance $1-\eta^\prime = \frac{1-\eta}{G}$ \cite{garciapatron2012majorization, caruso2006weakdegradability}. Using Lemma~\ref{lemma:pure-loss} in \ref{app:usefullemmas}, the system in  Fig.~\ref{fig:ent-break-loss} is equivalent to a single beamsplitter with reflectance $1-\nu = \tau(1-\eta^\prime)$ followed by the amplifier with gain $G$. The modal relationships of a quantum-limited amplifier with gain $G$ and input modes $\hat{f}$, $\hat{g}$ and output modes $\hat{h}$, $\hat{i}$ are $\hat{h} = \sqrt{G}\hat{f}+\sqrt{G-1}\hat{g}$ and  $\hat{i} = \sqrt{G-1}\hat{f}+\sqrt{G}\hat{g}$. These allow  for the expansion of the characteristic function for the output state:
\begin{align}
    \chi_A^{\hat{\rho}^W}\left(\zeta\right) &=\chi_A^{\hat{\rho}^A}\left(\sqrt{G(1-\nu)}\zeta\right) \chi_A^{\hat{\rho}^V}\left(\sqrt{G\nu}\zeta\right) \notag \\
    &\phantom{=}\times\chi_A^{\hat{\rho}^G}\left(\sqrt{(G-1)}\zeta\right) \notag\\
    &=\chi_A^{\hat{\rho}^A}\left(\sqrt{\tau(1-\eta)}\zeta\right)  e^{-(G-\tau(1-\eta))|\zeta|^2} e^{-(G-1)|\zeta|^2}\notag \\ 
    &= \chi_A^{\hat{\rho}^A}\left(\sqrt{\tau(1-\eta)}\zeta\right) e^{-(\frac{\eta\nb}{1-\tau(1-\eta)}+1)(1-\tau(1-\eta))|\zeta|^2}\notag
\end{align}
yielding a complementary channel with reflectance $\tau(1-\eta)$ and contribution of thermal mean photon number from the environment of $\eta\nb$ at the output. When $\tau < \frac{\eta\nb}{1-\eta}$, this meets the entanglement breaking condition $\eta\nb > \tau(1-\eta)$ \cite{holevoEntanglementbreakingChannelsInfinite2008}.


\subsection{Proof of Lemma \ref{lemma:ent-break-noise}} \label{proof-lemma-ent-noise}
\begin{figure}[htb]
\centering
\includegraphics[width=7.3cm]{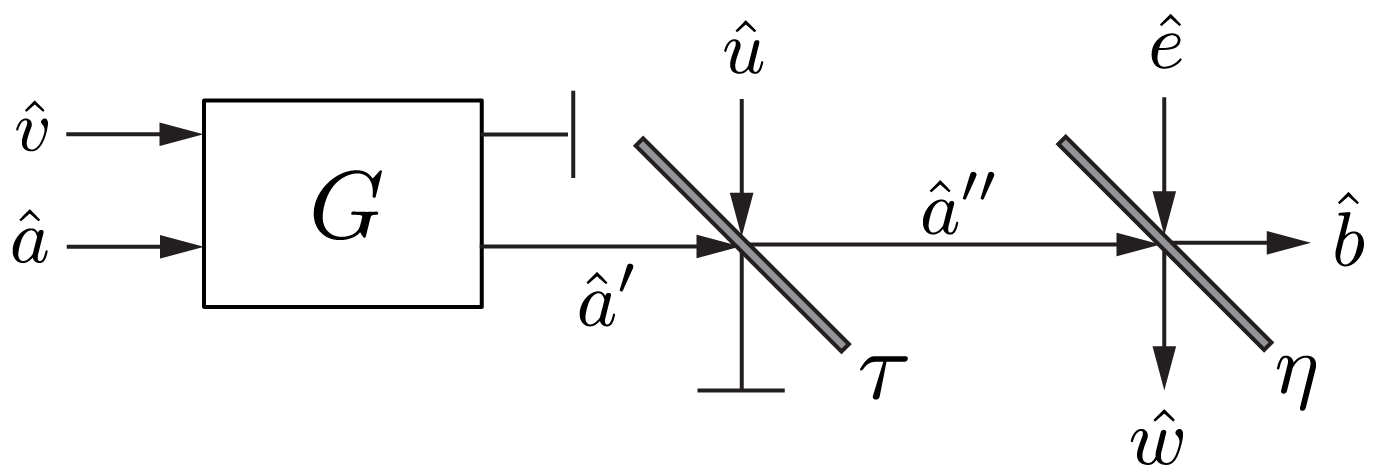}
\caption{A system that allows Alice to introduce additional arbitrary thermal noise in her channel to Willie. Annihilation operators $\hat{a}$ and $\hat{v}$ label Alice's input into the quantum-limited amplifier with gain coefficient $G$. The output of the amplifier is labeled by annihilation operator $\hat{a}^\prime$ and is input to a beamsplitter with tranmissivity $\tau$. The other port for the beamsplitter is labeled with $\hat{u}$. One output mode of the beamsplitter is traced out, while the other is labeled $\hat{a}^{\prime\prime}$. It is the input to a beamsplitter with transmissivity $\eta$ with a second input mode labeled $\hat{e}$, which represents the original channel in Fig.~\ref{fig:system} and is fixed. Alice has control of $G$, $\tau$, and the input states at modes $\hat{v}$ and $\hat{u}$ with her original qubit state input at $\hat{a}$. }
\label{fig:ent-break-noise}
\end{figure}

Alice utilizes the system devised in Fig.~\ref{fig:ent-break-noise} where she inputs her qubit state to a quantum-limited amplifier with gain coefficient $G$. The output of the amplifier is then fed to a beamsplitter of transmissivity $\tau$. Finally, the signal is passed through the channel modeled by a beamsplitter with transmissivity $\eta$. Alice has control over $G$, $\tau$, and the second input states to the amplifier and beamsplitter with tranmissivity $\tau$. $\eta$ and the environment input state at the mode labeled $\hat{e}$ are fixed. The input-output modal relationships of the system are
\begin{align}
\hat{a}^\prime &= \sqrt{G}\hat{a} + \sqrt{G-1}\hat{v},\quad
\hat{a}^{\prime\prime} = \sqrt{\tau}\hat{a}^\prime + \sqrt{1-\tau}\hat{u}\\
\hat{w} &= \sqrt{1-\eta}\hat{a}^{\prime\prime} -\sqrt{\eta}\hat{e} \\
 &= \sqrt{\tau(1-\eta)}\hat{a}^\prime + \sqrt{(1-\tau)(1-\eta)}\hat{u} -\sqrt{\eta}\hat{e} \\
 &= \sqrt{G\tau(1-\eta)}\hat{a} + \sqrt{(G-1)\tau(1-\eta)}\hat{v} \notag \\
& \phantom{=}+ \sqrt{(1-\tau)(1-\eta)}\hat{u} -\sqrt{\eta}\hat{e} \label{eq:ent-relationship}
\end{align}

To demonstrate the state at Willie is equivalent to the original channel with additional noise, we utilize the characteristic functions of the system. Without loss of generality, we consider the single-mode case. Thus, we need to show:
$\chi_A^{\hat{\rho}^W}\left(\zeta\right) =\chi_A^{\hat{\rho}^A}\left(\sqrt{1-\eta}\zeta\right) e^{-(1+ \eta(\nb + \nb^\prime))|\zeta|^2}$, where $\nb$ is the thermal mean photon number of the environment state and $\nb^\prime$ is the additional thermal mean photon number Alice introduces. The exponential is of the same form as the single mode version of \eqref{eq:char-func-e-def}. It follows then that
\begin{align}
    \chi_A^{\hat{\rho}^W}\left(\zeta\right) &= \chi_A^{\hat{\rho}^A}\left(\sqrt{G\tau(1-\eta)}\zeta\right)    \notag \\
    & \phantom{=}\times \chi_A^{\hat{\rho}^V}\left(\sqrt{(G-1)\tau(1-\eta)}\zeta\right) \notag \\
    & \phantom{=}\times\chi_A^{\hat{\rho}^U}\left(\sqrt{(1-\tau)(1-\eta)}\zeta\right)  \chi_A^{\hat{\rho}^E}\left(\sqrt{\eta}\zeta\right) \label{eq:ent-proof-noise-1} \\
    & = \chi_A^{\hat{\rho}^A}\left(\sqrt{1-\eta}\zeta\right)     \chi_A^{\hat{\rho}^V}\left(\sqrt{(1-\frac{1}{G})(1-\eta)}\zeta\right) \notag \\
    & \phantom{=}\times\chi_A^{\hat{\rho}^U}\left(\sqrt{(1-\frac{1}{G})(1-\eta)}\zeta\right)  \chi_A^{\hat{\rho}^E}\left(\sqrt{\eta}\zeta\right) \label{eq:ent-proof-noise-2}\\
    & = \chi_A^{\hat{\rho}^A}\left(\sqrt{1-\eta}\zeta\right)\notag \\
    & \phantom{=}\times e^{-(1-\frac{1}{G})(1-\eta)|\zeta|^2}e^{-(1-\frac{1}{G})(1-\eta)|\zeta|^2}e^{-(1+ \nb)\eta|\zeta|^2} \label{eq:ent-proof-noise-3}\\
    & = \chi_A^{\hat{\rho}^A}\left(\sqrt{1-\eta}\zeta\right) e^{-(1+ \eta(\nb + \frac{2}{\eta}(1-\eta)(1-\frac{1}{G}))|\zeta|^2}\label{eq:ent-proof-noise-4} \\
    & = \chi_A^{\hat{\rho}^A}\left(\sqrt{1-\eta}\zeta\right) e^{-(1+ (\nb + \nb^\prime))\eta)|\zeta|^2} \label{eq:ent-proof-noise-5}.
\end{align}
As the inputs to the system are not entangled, in \eqref{eq:ent-proof-noise-1} we split the characteristic function into input components described by \eqref{eq:ent-relationship}. In \eqref{eq:ent-proof-noise-2}, we set $\tau = 1/G$, and in \eqref{eq:ent-proof-noise-3} we evaluate $\chi_A^{\hat{\rho}^V}$, $\chi_A^{\hat{\rho}^U}$, and $\chi_A^{\hat{\rho}^E}$ for respective input states of $\hat{\rho}^V=\hat{\rho}^U=|0\rangle\langle0|$ and $\hat{\rho}^E=\hat{\rho}_{\nb}$ is a thermal state. In \eqref{eq:ent-proof-noise-4} and \eqref{eq:ent-proof-noise-5} we show that the characteristic function of Willie's state is equivalent to that of our original channel with additional thermal noise $\nb^\prime = \frac{2}{\eta}(1-\eta)(1-\frac{1}{G})$ that meets the entanglement breaking condition if $\eta(\bar{n}_{\textrm{B}}+\bar{n}_{\textrm{B}}^\prime) > 1-\eta$. Solving for $G$ yields $G=(1-\eta)/((1-\eta)-\eta \frac{\nb^\prime}{2})$ and completes the proof.


\section{Proof of Lemma 3}\label{sec:proof-of-lemma3}

The proof is structured as follows: we derive Willie's state when Alice transmits, and then compute the quantum $\chi^2$-divergence between it and Willie's state when Alice does not transmit.
We use Lemmas \ref{lemma:integral-theta} and \ref{lemma:integral-theta2} from \ref{app:usefullemmas}.
\subsection{State at Willie}
When Alice chooses to transmit, her input $\hat{\rho}^A$ is given by \eqref{eq:alice-arb-input}. This is mixed with the thermal state $\hat{\rho}^E=\hat{\rho}_{\nb}$ by the channel in Fig.~\ref{fig:system}.  
The input-output relationship between modal annihilation  operators  in Fig.~\ref{fig:system} yields 
$\hat{w}=\sqrt{1-\eta} \hat{a}-\sqrt{\eta} \hat{e} \label{eq:modal-relationship}$.

Let us now find the Fock representation of the output state $\hat{\rho}^W$ of Warden when Alice transmits a qubit. We utilize the anti-normally ordered characteristic function for $\hat{\rho}^W$:
\begin{align}
\chi_A^{\hat{\rho}^W}\left(\zeta_1,\zeta_2\right) &= \tr(\hat{\rho}^W e^{-\zeta_1^* \hat{w}_1} e^{\zeta_1 \hat{w}_1^{\dagger}} e^{-\zeta_2^* \hat{w}_2} e^{\zeta_2 \hat{w}_2^{\dagger}}).
\end{align}
Furthermore, let $\langle\cdot\rangle_x$ denote the expected value of an operator when applied to the state $\hat{\rho}^x$ for $x \in \{A,W,E\}$,  we have
\begin{align}
&\chi_A^{\hat{\rho}^W}\left(\zeta_1,\zeta_2\right) \notag\\ 
&= \langle e^{-\zeta_1^* \hat{w}_1} e^{\zeta_1 \hat{w}_1^{\dagger}} e^{-\zeta_2^* \hat{w}_2} e^{\zeta_2 \hat{w}_2^{\dagger}} \rangle_W \\
&=\langle e^{-\zeta_1^* (\sqrt{1-\eta} \hat{a}_1-\sqrt{\eta} \hat{e}_1)} e^{\zeta_1 (\sqrt{1-\eta} \hat{a}_1^\dagger-\sqrt{\eta} \hat{e}_1^\dagger)} \notag \\
&\phantom{=\langle} \times e^{-\zeta_2^* (\sqrt{1-\eta} \hat{a}_2-\sqrt{\eta} \hat{e}_2)} e^{\zeta_2 (\sqrt{1-\eta} \hat{a}_2^\dagger-\sqrt{\eta} \hat{e}_2^\dagger)} \rangle_W \label{eq:subW_hat} \\
&=\langle e^{-\zeta_1^* \sqrt{1-\eta} \hat{a}_1} e^{\zeta_1 \sqrt{1-\eta} \hat{a}_1^\dagger}e^{-\zeta_1^* \sqrt{1-\eta} \hat{a}_2} \notag \\
&\phantom{=\langle} \times e^{\zeta_2 \sqrt{1-\eta} \hat{a}_2^\dagger}e^{\zeta_1^*\sqrt{\eta} \hat{e}_1}e^{-\zeta\sqrt{\eta} \hat{e}_1^\dagger}e^{\zeta_2^*\sqrt{\eta} \hat{e}_2}e^{-\zeta \sqrt{\eta} \hat{e}_2^\dagger}\rangle_w \label{eq:CommuteModes1&2} \\
&=\chi_A^{\hat{\rho}^A}\left(\sqrt{1-\eta}\zeta_1,\sqrt{1-\eta}\zeta_2\right) \chi_A^{\hat{\rho}^E}\left(-\sqrt{\eta}\zeta_1,-\sqrt{\eta}\zeta_2\right) \label{eq:seperateCharFunctions}
\end{align}
where \eqref{eq:subW_hat} uses $\hat{w}=\sqrt{1-\eta}\hat{a}-\sqrt{\eta}\hat{e}$, \eqref{eq:CommuteModes1&2} uses the fact that bosonic operators acting on different modes commute with each other, and \eqref{eq:seperateCharFunctions} separates the characteristic function into characteristic functions for each input mode of Alice $\chi_A^{\hat{\rho}^A}$, and the environment $\chi_A^{\hat{\rho}^E}$. 

To derive the characteristic function for mode $\hat{e}$, note that the multi-mode thermal state $\hat{\rho}^E$ is separable, and thus
$\chi_A^{\hat{\rho}^E}(-\sqrt{\eta}\zeta_1, -\sqrt{\eta}\zeta_2)= \chi_A^{\hat{\rho}^E}(-\sqrt{\eta}\zeta_1) \chi_A^{\hat{\rho}^E}(-\sqrt{\eta}\zeta_2)$.
The solution for the single mode $\chi_A^{\hat{\rho}_{e}}$ is well known \cite[Sec.~7.4.3.2]{orszag16quantumotpics}:
\begin{align}
    \chi_A^{\hat{\rho}^E}(-\sqrt{\eta}\zeta_1, -\sqrt{\eta}\zeta_2) &= e^{-(1+\nb)\eta(|\zeta_1|^2+|\zeta_2|^2)}. \label{eq:char-func-e-def}
\end{align}

Now, we calculate the characteristic function for $\hat{\rho}^A$:
\begin{align}
    &\chi_A^{\hat{\rho}^A}\left(\sqrt{1-\eta}\zeta_1,\sqrt{1-\eta}\zeta_2\right) \label{eq:char-func-a-def} \\
    &= e^{-(1-\eta)(|\zeta_1|^2+|\zeta_2|^2)} \notag \\ 
    &\phantom{=} \times\tr\left(\hat{\rho}^A e^{\zeta_1 \sqrt{1-\eta} \hat{a}_1^\dagger} e^{-\zeta_1^* \sqrt{1-\eta} \hat{a}_1} e^{\zeta_2 \sqrt{1-\eta} \hat{a}_2^\dagger} e^{-\zeta_2^* \sqrt{1-\eta} \hat{a}_2} \right) \label{eq:BCHtheorem}\\
    &=e^{-(1-\eta)(|\zeta_1|^2+|\zeta_2|^2)}\notag \\
    &\phantom{=} \times \sum_{n=0}^\infty \sum_{m=0}^\infty  \langle nm| \left( |\alpha|^2 |01\rangle\langle01|+\gamma |01\rangle\langle10| \right . \notag \\
    &\phantom{=\sum}\left.+\gamma^*|10\rangle\langle01|+|\beta|^2|10\rangle\langle10| \right) \notag \\
    & \phantom{=\sum}  \times e^{\zeta_1 \sqrt{1-\eta} \hat{a}_1^\dagger} e^{-\zeta_1^* \sqrt{1-\eta} \hat{a}_1} e^{\zeta_2 \sqrt{1-\eta} \hat{a}_2^\dagger} e^{-\zeta_2^* \sqrt{1-\eta} \hat{a}_2} |nm\rangle \notag
\end{align}
where \eqref{eq:BCHtheorem} is due to the Baker-Campbell-Hausdorff theorem, followed by expansion of the trace. Furthermore,
\begin{subequations}
    \begin{align}
    &\chi_A^{\hat{\rho}^A}\left(\sqrt{1-\eta}\zeta_1,\sqrt{1-\eta}\zeta_2\right) \notag \\
    &= e^{-(1-\eta)(|\zeta_1|^2+|\zeta_2|^2)} \notag \\
    &\phantom{=}\times\big( |\alpha|^2\langle 01| \hat{V}(\zeta_1,\zeta_2, \hat{a}_1,\hat{a}_2)|01\rangle \label{eq:alpha2_BraKet} \\
    &\phantom{=} + \gamma  \langle 10| \hat{V}(\zeta_1,\zeta_2, \hat{a}_1,\hat{a}_2) |01\rangle \label{eq:alpha_betastar_BraKet} \\
    &\phantom{=} + \gamma^* \langle 01| \hat{V}(\zeta_1,\zeta_2, \hat{a}_1,\hat{a}_2)|10\rangle \label{eq:alphastar_beta_BraKet} \\
    &\phantom{=} + |\beta|^2\langle 10| \hat{V}(\zeta_1,\zeta_2, \hat{a}_1,\hat{a}_2) |10\rangle \big) \label{eq:beta2_BraKet}
\end{align}
\end{subequations}
where, for compactness, we define $\hat{V}(\zeta_1,\zeta_2, \hat{a}_1,\hat{a}_2)  \equiv e^{\zeta_1 \sqrt{1-\eta} \hat{a}_1^\dagger}e^{-\zeta_1^* \sqrt{1-\eta} \hat{a}_1} e^{\zeta_2 \sqrt{1-\eta} \hat{a}_2^\dagger} e^{-\zeta_2^* \sqrt{1-\eta} \hat{a}_2}$.
Expanding the exponentials for mode 1 of \eqref{eq:alpha2_BraKet} yields:
\begin{align}
    &e^{-(1-\eta)(|\zeta_1|^2+|\zeta_2|^2)} \notag \\
    &\times|\alpha|^2\langle 01| e^{\zeta_1 \sqrt{1-\eta} \hat{a}_1^\dagger}e^{-\zeta_1^* \sqrt{1-\eta} \hat{a}_1} e^{\zeta_2 \sqrt{1-\eta} \hat{a}_2^\dagger} e^{-\zeta_2^* \sqrt{1-\eta} \hat{a}_2}  |01\rangle\notag \\
    &= e^{-(1-\eta)(|\zeta_1|^2+|\zeta_2|^2)} |\alpha|^2\langle 01|\left(\sum_{k^{\prime}=0}^n(\hat{a}_1^\dagger)^{k^\prime}\frac{(\sqrt{1-\eta} \zeta)^{k^{\prime}}}{k^\prime!} \right)\notag \\
    &\phantom{=} \times \left(\sum_{k=0}^n\frac{\left(-\sqrt{1-\eta} \zeta^*\right)^k}{k!} (\hat{a}_1)^k\right) e^{-\zeta_2^* \sqrt{1-\eta} \hat{a}_2} e^{\zeta_2 \sqrt{1-\eta} \hat{a}_2^\dagger} |01\rangle.\notag
\end{align}
As mode 1 is in the $0$ state, every term in the summations is zero unless $k^\prime$ and $k$ are equal to 0. Thus, \eqref{eq:alpha2_BraKet} reduces to: 
\begin{align}
    &e^{-(1-\eta)(|\zeta_1|^2+|\zeta_2|^2)} |\alpha|^2 \langle 01| e^{-\zeta_2^* \sqrt{1-\eta} \hat{a}_2} e^{\zeta_2 \sqrt{1-\eta} \hat{a}_2^\dagger} |01\rangle. \label{eq:expo-mode-2}
\end{align}
Expanding the exponentials for mode 2 in \eqref{eq:expo-mode-2} yields:
\begin{align}
    &e^{-(1-\eta)(|\zeta_1|^2+|\zeta_2|^2)} |\alpha|^2 \bra{01}\left(\sum_{k^{\prime}=0}^n(\hat{a}_2^\dagger)^{k^\prime}\frac{(\sqrt{1-\eta} \zeta)^{k^{\prime}}}{k^\prime!} \right)\notag \\
    &\times\left(\sum_{k=0}^n\frac{\left(-\sqrt{1-\eta} \zeta^*\right)^k}{k!} (\hat{a}_2)^k\right)  \ket{01} \\
    &= e^{-(1-\eta)(|\zeta_1|^2+|\zeta_2|^2)} |\alpha|^2 \langle 0| \otimes (\langle1|+(\sqrt{1-\eta}\zeta_2)\langle0|)|0\rangle \notag \\
    &\phantom{=}\otimes (|1\rangle-\zeta_2^*\sqrt{1-\eta}|0\rangle) \\
    &= e^{-(1-\eta)(|\zeta_1|^2+|\zeta_2|^2)} |\alpha|^2(1-(1-\eta)|\zeta_2|^2).
\end{align}
\crefrange{eq:alpha_betastar_BraKet}{eq:beta2_BraKet} follow similarly to yield:
\begin{subequations}
\begin{align}
    &\chi_A^{\hat{\rho}^W}(\zeta_1,\zeta_2) \notag \\
    &= e^{-(1+\eta \nb)(|\zeta_1|^2+|\zeta_2|^2)}|\alpha|^2(1-(1-\eta)|\zeta_2|^2) \\
    &\phantom{=} -e^{-(1+\eta \nb)(|\zeta_1|^2+|\zeta_2|^2)}\gamma  (1-\eta) \zeta_1\zeta_2^* \\
    &\phantom{=} - e^{-(1+\eta \nb)(|\zeta_1|^2+|\zeta_2|^2)}\gamma^* (1-\eta)\zeta_1^*\zeta_2 \\
    &\phantom{=} + e^{-(1+\eta \nb)(|\zeta_1|^2+|\zeta_2|^2)}|\beta|^2 (1-(1-\eta)|\zeta_1|^2),
\end{align}
\end{subequations}
which simplifies to \eqref{eq:charfuncwillie-simplified}. A generalized two-mode state is:
\begin{align}
    \hat{\rho}^W = \sum_{m=0}^\infty\sum_{n=0}^\infty\sum_{m^\prime=0}^\infty\sum_{n^\prime=0}^\infty \langle nm|\hat{\rho}^W|n^\prime m^\prime\rangle |nm\rangle\langle n^\prime m^\prime|. \label{eq:gen-state}
\end{align}
Applying the operator Fourier transform in \eqref{eq:fourier-transform} on the characteristic function $\chi_A^{\hat{\rho}^W}(\zeta_1,\zeta_2)$, we obtain: 
\begin{subequations}
\begin{align}
    \hat{\rho}^W &= \sum_{m=0}^\infty\sum_{n=0}^\infty\sum_{m^\prime=0}^\infty\sum_{n^\prime=0}^\infty \langle nm| \notag\\
    &\phantom{=}\left ( |\alpha|^2 \int \frac{d^2\zeta_2}{\pi} \int \frac{d^2\zeta_1}{\pi} \kappa_1\kappa_2(1-(1-\eta)|\zeta_2|^2)   \right. \\
    &\phantom{=} + \gamma  \int \frac{d^2\zeta_2}{\pi} \int \frac{d^2\zeta_1}{\pi} \kappa_1\kappa_2(-(1-\eta)\zeta_1\zeta_2^*)   \\
    &\phantom{=} + \gamma^*  \int \frac{d^2\zeta_2}{\pi} \int \frac{d^2\zeta_1}{\pi} \kappa_1\kappa_2 (-(1-\eta)\zeta_1^*\zeta_2)   \\
    &\phantom{=} \left . + |\beta|^2 \int \frac{d^2\zeta_2}{\pi} \int \frac{d^2\zeta_1}{\pi} \kappa_1\kappa_2 (1-(1-\eta)|\zeta_1|^2)    \right ) \\ 
    &\phantom{=} |n^\prime m^\prime\rangle |nm\rangle\langle n^\prime m^\prime|,
\end{align}  
\end{subequations}
where $\kappa_i \equiv e^{-(1+\eta \nb)|\zeta_i|^2} e^{\zeta_i\hat{c}_i^\dagger} e^{-\zeta_i^*\hat{c}_i}$ and $i=1,2$.
The integrals are separable with respect to each mode, yielding: 
\begin{subequations}
\begin{align}
    \hat{\rho}^W &= \sum_{m=0}^\infty\sum_{n=0}^\infty\sum_{m^\prime=0}^\infty\sum_{n^\prime=0}^\infty \langle nm|  \notag\\
    &\phantom{=} \left ( |\alpha|^2  \int\frac{d^2\zeta_2}{\pi} \kappa_2 (1-(1-\eta)|\zeta_2|^2) 
    \int  \frac{d^2\zeta_1}{\pi} \kappa_1 \label{eq:sep_alpha2} \right .\\
    &\phantom{=} - \gamma (1-\eta)  \int \frac{d^2\zeta_2}{\pi} \kappa_2 \zeta_2^*  \int  \frac{d^2\zeta_1}{\pi} \kappa_1 \zeta_1   \label{eq:sep_alpha_betastar} \\
    &\phantom{=} - \gamma^*  (1-\eta) \int \frac{d^2\zeta_2}{\pi} \kappa_2 \zeta_2  \int  \frac{d^2\zeta_1}{\pi} \kappa_1 \zeta_1^*  \label{eq:sep_alphastar_beta} \\
    &\phantom{=} \left . + |\beta|^2  \int  \frac{d^2\zeta_2}{\pi} \kappa_2 \int    \frac{d^2\zeta_1}{\pi} \kappa_1 (1-(1-\eta)|\zeta_1|^2) \label{eq:sep_beta2} \right ) \\ 
    &\phantom{=} |n^\prime m^\prime\rangle |nm\rangle\langle n^\prime m^\prime|
\end{align}
\end{subequations}

Next, we calculate $\langle nm|\hat{\rho}^W|n^\prime m^\prime\rangle$ for each term in \crefrange{eq:sep_alpha2}{eq:sep_beta2}. We first convert to polar coordinates and use Lemma \ref{lemma:integral-theta} to integrate over $\theta$. Equation \eqref{eq:sep_alpha2} simplifies to:
\begin{align}
    &4|\alpha|^2\int_{r_2=0}^\infty dr_2 r_2 (1-(1-\eta) r_2^2) e^{-(1+\eta \nb)r_2^2} \mathcal{L}_m(r_2^2)  \notag \\
    &\times\int_{r_1=0}^\infty dr_1 r_1 e^{-(1+\eta \nb)r_1^2} \mathcal{L}_n(r_1^2).  \label{eq:tequation1}
\end{align}
These integrals are further evaluated to the following:
\begin{align}
    &|\alpha|^2 \left(\frac{(\eta \nb)^m}{(1+\eta \nb)^{m+1}}- \frac{(1-\eta)(\eta \nb-m)(\eta \nb)^{m-1}}{(1+\eta \nb)^{m+2}}  \right) \notag \\
    &\times \frac{(\eta \nb)^n}{(1+\eta \nb)^{n+1}}. \label{eq:aa-sol}
\end{align}
The similar integrals in \eqref{eq:sep_beta2} follow with $m$ and $n$ swapped:
\begin{align}
    &|\beta|^2 \left(\frac{(\eta \nb)^n}{(1+\eta \nb)^{n+1}}- \frac{(1-\eta)(\eta \nb-n)(\eta \nb)^{n-1}}{(1+\eta \nb)^{n+2}}  \right)\notag \\
    &\times \frac{(\eta \nb)^m}{(1+\eta \nb)^{m+1}}. \label{eq:bb-sol}
\end{align}

For \eqref{eq:sep_alpha_betastar}, we use polar coordinates and Lemma~\ref{lemma:integral-theta2}:
\begin{align}
&- \gamma (1-\eta)  \int \frac{d^2\zeta_2}{\pi} \kappa_2 \zeta_2^*  \int  \frac{d^2\zeta_1}{\pi} \kappa_1 \zeta_1\notag \\
&= - 4\gamma (1-\eta)\sqrt{m+1}\sqrt{n^\prime+1} \int_{r_2=0}^\infty dr_2 r_2^3 e^{-(1+\eta \nb)r_2^2} \notag \\
&\phantom{=}\times\sum_{k=0}^{m} (r_2^2)^{k}\frac{(-1)^{k}}{k!}\binom{m}{k}\frac{1}{k+1}  \notag \\
&\phantom{=} \times \int_{r_1=0}^\infty dr_1 r_1^3 e^{-(1+\eta \nb)r_1^2}  \sum_{l=0}^{n^\prime} (r_1^2)^{l}\frac{(-1)^{l}}{l!}\binom{n^\prime}{l}\frac{1}{l+1}\notag \\
&= - \gamma (1-\eta)\sqrt{m+1}\sqrt{n^\prime+1}   \sum_{k=0}^{m} \frac{(k+1)!}{(1+\eta \nb)^{k+2}} \frac{(-1)^{k}}{k!}  \notag \\
&\phantom{=} \times  \binom{m}{k}\frac{1}{k+1}\sum_{l=0}^{n^\prime} \frac{(l+1)!}{(1+\eta \nb)^{l+2}}\frac{(-1)^{l}}{l!}\binom{n^\prime}{l}\frac{1}{l+1} \notag \\
&= - \gamma (1-\eta)\sqrt{m+1}\sqrt{n^\prime+1} \notag \\
&\phantom{=}\times  \sum_{k=0}^{m} \frac{(-1)^k}{(1+\eta \nb)^{k+2}} \binom{m}{k} \sum_{l=0}^{n^\prime} \frac{(-1)^{l}}{(1+\eta \nb)^{l+2}}\binom{n^\prime}{l} \notag\\
&= - \gamma (1-\eta)\sqrt{m+1}\sqrt{n^\prime+1} \notag \\
&\phantom{=}\times\frac{(\eta \nb)^m}{(1+\eta \nb)^{m+2}} \frac{(\eta \nb)^{n^\prime}}{(1+\eta \nb)^{n^\prime+2}} \notag\\
&= - \gamma (1-\eta)\sqrt{m+1}\sqrt{n} \frac{(\eta \nb)^{m+n-1}}{(1+\eta \nb)^{m+n+3}},  \label{eq:ab-int-polar4}
\end{align}
where we use $n^\prime=n-1$ for compactness.

The integrals in \eqref{eq:sep_alphastar_beta} follow the same steps to yield:
\begin{align}
        &\phantom{=} - \gamma^*  (1-\eta) \int \frac{d^2\zeta_2}{\pi} \kappa_2 \zeta_2  \int  \frac{d^2\zeta_1}{\pi} \kappa_1 \zeta_1^* \notag \\
        &=-\gamma^*  (1-\eta) \sqrt{m}\sqrt{n+1}\frac{(\eta \nb)^{m+n-1}}{(1+\eta \nb)^{m+n+3}} \label{eq:astarbeta-sol}.
\end{align}

Replacing the $\langle nm|\hat{\rho}^W|n^\prime m^\prime\rangle$ in \eqref{eq:gen-state} with \eqref{eq:aa-sol}-\eqref{eq:astarbeta-sol} yields \eqref{eq:SumTrace}.
The $\langle n+1,m-1|$ and $\langle n-1,m+1|$ terms in \eqref{eq:SumTrace} are due to the off-by-one integrals over $\theta$ as per Lemma~\ref{lemma:integral-theta2}.

In the rounds when Alice does not transmit, Willie observes a dual-mode attenuated thermal state $\hat{\rho}_0\equiv\left(\hat{\rho}_{\eta\nb}\right)^{\otimes 2}$. 

\subsection{Calculating the \texorpdfstring{$\chi^2$}{chi-squared}-divergence}
We seek $D_{\chi^2}(\hat{\rho}^W\|\hat{\rho}_0)$ defined in \eqref{eq:chisquare}, where $\hat{\rho}^W$ is in \eqref{eq:SumTrace}. 
Since $\hat{\rho}_0$ is diagonal, $\hat{\rho}_0^{-1}$ is also diagonal. Consequently, only the diagonal elements of $\left(\hat{\rho}^W\right)^2$ are required to evaluate $\tr\left[\left(\hat{\rho}^W\right)^2\hat{\rho}_{0}^{-1}\right]$. As $\hat{\rho}^W$ is a tri-diagonal state, is decomposes into operators $\hat{A}$, $\hat{B}$, and $\hat{D}$ corresponding to each of the diagonals: $\hat{\rho}^W = \hat{A}+\hat{B}+\hat{D}$, where 
\begin{align}
    \hat{A} &= \sum_{n=0}^\infty \sum_{m=1}^\infty a_{nm} |nm\rangle\langle n+1,m-1| \label{eq:A} \\
    \hat{B} &= \sum_{n=1}^\infty \sum_{m=0}^\infty b_{nm} |nm\rangle\langle n-1,m+1|, \label{eq:B}\\
   \hat{D} &= \sum_{n=0}^\infty \sum_{m=0}^\infty d_{nm} |nm\rangle\langle nm|, \label{eq:D}
\end{align}
and the coefficients are $a_{nm} = -\gamma^*  W_2(n,m)$, $b_{nm} = a_{nm}^*$, and $d_{nm} = |\alpha|^2 W_1(n, m) + |\beta|^2 W_1(m, n)$.

We then expand $\left(\hat{\rho}^W\right)^2$ as
\begin{align}
    \left(\hat{\rho}^W\right)^2 &= (\hat{A}+\hat{B}+\hat{D})^2 \\ 
    &= \hat{A}^2+\hat{A}\hat{B}+\hat{A}\hat{D}+\hat{B}\hat{A}\notag\\
    &\phantom{=}+\hat{B}^2+\hat{B}\hat{D}+\hat{D}\hat{A}+\hat{D}\hat{B}+\hat{D}^2
\end{align}
By substitution from \crefrange{eq:A}{eq:D} and orthonormality of the Fock states, we obtain:
\begin{align}
    \hat{A}^2 &= \sum_{n=0}^\infty \sum_{m=2}^\infty a_{nm} a_{n+1,m-1} |nm\rangle\langle n+2,m-2| \\
    \hat{B}^2 &= \sum_{n=2}^\infty \sum_{m=0}^\infty b_{nm} b_{n-1,m+1} |nm \rangle \langle n-2,m+2| \\
    \hat{D}^2 &= \sum_{n=0}^\infty \sum_{m=0}^\infty d_{nm}^2 |nm \rangle \langle nm| \\
    \hat{A}\hat{B} &= \sum_{n=0}^\infty \sum_{m=1}^\infty a_{nm}b_{n+1,m-1} |nm \rangle \langle nm| \\
    \hat{B}\hat{A} &= \sum_{n=1}^\infty \sum_{m=0}^\infty a_{n-1,m+1}b_{nm} |nm \rangle \langle nm| \\
    \hat{A}\hat{D} = \hat{D}\hat{A} &= \sum_{n=0}^\infty \sum_{m=1}^\infty a_{nm}d_{n+1,m-1} |nm\rangle\langle n+1,m-1| \\
    \hat{B}\hat{D} = \hat{D}\hat{B} &= \sum_{n=1}^\infty \sum_{m=0}^\infty b_{nm} d_{n-1,m+1} |nm\rangle\langle n-1,m+1|.
\end{align}

Hence, the only terms that contribute to the main diagonal of $\left(\hat{\rho}^W\right)^2$ are $\hat{D}^2$, $\hat{A}\hat{B}$, and $\hat{B}\hat{A}$. Therefore,
\begin{align}
    &D_{\chi^2}(\hat{\rho}^W\|\hat{\rho}_0) = \tr[\left(\hat{\rho}^W\right)^{2}\hat{\rho}_{0}^{-1}] - 1 \\
    &=  \sum_{n=0,m=0}^\infty (d_{nm}^2 + a_{nm}b_{n+1,m-1} + a_{n-1,m+1}b_{nm}) t_{nm}^{-1} -1 \label{eq:chi-square-f2} \\
&= -1 + \left[\left((1-\eta)^2+\eta\nb(1+\eta\nb)\right)(|\alpha|^4+|\beta|^4) \right. \notag \\
&\phantom{=} \left. + 2\left(|\alpha|^2|\beta|^2\eta\nb(1+\eta\nb)+ (1-\eta)^2|\gamma|^4\right)\right] \notag \\
&\phantom{=}\times\frac{1}{\eta\nb(1+\eta\nb)} \label{eq:chi-square-f3}.
\end{align}
where $t_{nm}^{-1}\equiv \left(t_n(\eta\nb)t_m(\eta\nb)\right)^{-1}$ , with $t_k(\cdot)$ defined in \eqref{eq:thermCoefficient}. 
We obtain \eqref{eq:chi-square-f3} using standard algebra and known closed forms of the summations. Using the constraint $|\alpha|^2 + |\beta|^2 = (|\alpha|^2 + |\beta|^2)^2 = 1$ and $|\gamma| < |\alpha\beta|$ shows that qubit input, where $\gamma=\alpha\beta^*$, maximizes the $\chi^2$-divergence and yields \eqref{eq:chisquared-qubit} with the lemma.
Furthermore, the $\chi^2$-divergence is minimized by a balanced classical mixture with $\gamma = 0$, $|\alpha|^2 = |\beta|^2=1/2$, and is equal to $1/2$ of the expression in \eqref{eq:chisquared-qubit}.

\renewcommand{\thesection}{Appendix \Roman{section}}
\section{Useful Lemmas}
\label{app:usefullemmas}
\renewcommand{\thesection}{\Roman{section}}
\begin{lemma} \label{lemma:pure-loss}
    The output of complementary channel for a pure-loss bosonic channel with transmittance $\tau$ followed by one with transmittance $\eta$ is equivalent to the output of a single complementary pure-loss bosonic channel with reflectance $\tau(1-\eta)$.
\end{lemma}
\begin{proof}
\begin{figure}[htb]
\centering
\includegraphics[width=7.5cm]{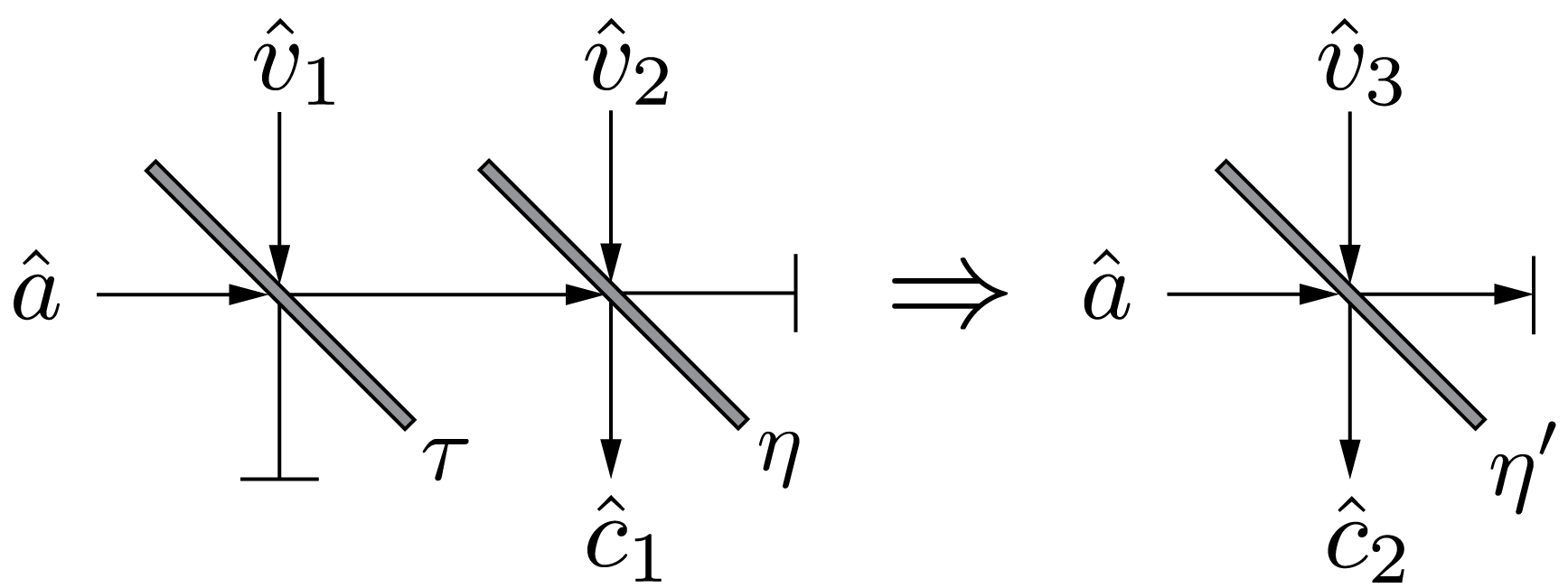}
\caption{The complementary channel formed by two pure-loss bosonic channels with transmittances $\eta$ and $\tau$ is equivalent to the complementary pure-loss bosonic channel with transmittance $\eta^\prime =1- \tau(1-\eta)$. Modes $\hat{v}_1,\hat{v}_2,\hat{v}_3$ are in vacuum state. 
}
\label{fig:decomposition}
\end{figure}
Consider the complementary channel to mode $\hat{c}_2$ on the right-hand side of Fig.~\ref{fig:decomposition} where mode $\hat{v}_3$ is in a vacuum state. The characteristic function for the state of mode $\hat{c}_2$ is:
\begin{align}
    \chi_A^{\hat{\rho}_{c_2}}(\zeta) &= \chi_A^{\hat{\rho}^A}(\sqrt{1-\eta^\prime}\zeta) \chi_A^{\hat{\rho}^{V_2}} (\sqrt{\eta^\prime}\zeta) \\
    &= \chi_A^{\hat{\rho}^A}(\sqrt{1-\eta^\prime)}\zeta) e^{-\eta^\prime|\zeta|^2} \label{eq:c2-state}
\end{align}
The characteristic function for the state of mode $\hat{c}_2$ is:
\begin{align}
 \chi_A^{\hat{\rho}_{c_1}}(\zeta) &= \chi_A^{\hat{\rho}^A}(\sqrt{\tau(1-\eta)}\zeta) \chi_A^{\hat{\rho}^{V_1}} (\sqrt{(1-\tau)(1-\eta)}\zeta) \notag \\
    & \phantom{=}\times\chi_A^{\hat{\rho}^{V_2}}(\sqrt{\eta}\zeta) \\
    &= \chi_A^{\hat{\rho}^A}(\sqrt{\tau(1-\eta)}\zeta) e^{-(1-\tau)(1-\eta)|\zeta|^2} e^{-\eta|\zeta|^2} \\
    &= \chi_A^{\hat{\rho}^A}(\sqrt{\tau(1-\eta)}\zeta) e^{-(1-\tau(1-\eta))|\zeta|^2} \\
    &= \chi_A^{\hat{\rho}^A}(\sqrt{1-\eta^\prime)}\zeta) e^{-\eta^\prime|\zeta|^2} \label{eq:c1-state}.
\end{align}
Therefore, the states in modes $\hat{c}_1$ and $\hat{c}_2$ are equivalent as \eqref{eq:c2-state} is equal to that of \eqref{eq:c1-state} completing the proof.
\end{proof}
\begin{lemma} \label{lemma:integral-theta}
\begin{align}
&\int_{\theta=0}^{2\pi}d\theta \langle m| e^{r(\cos(\theta) + \im\sin(\theta)) \hat{c}^{\dagger}} e^{-r(\cos(\theta) - \im\sin(\theta)) \hat{c}}  | m^\prime\rangle \notag \\ 
&=2\pi\mathcal{L}_m(r^2).
\end{align}
\end{lemma}

\begin{proof}
\begin{align}
&\int_{\theta=0}^{2\pi}d\theta \langle m| e^{r(\cos(\theta) + \im\sin(\theta)) \hat{c}^{\dagger}} e^{-r(\cos(\theta) - \im\sin(\theta)) \hat{c}}  | m^\prime\rangle \label{eq:integral_theta1} \\
&= \int_{\theta=0}^{2\pi} d\theta \sum_{k^{\prime}=0}^{\infty}\langle m|\frac{r^{k^\prime}e^{j\theta k^{\prime}}}{k^\prime!}(\hat{c}^\dagger)^{k^\prime}\sum_{k=0}^{\infty} \frac{(-r)^ke^{-j\theta k}}{k!}\hat{c}^{k}|m\rangle \label{eq:integral_theta2} \\
&= 2 \pi \sum_{k=0}^{m} \frac{(-r^2)^k}{k!}\sqrt{\frac{1}{k!^2} \frac{m!}{(m-k)!} \frac{m^\prime!}{(m^\prime-k)!} }  \langle m-k|m^\prime-k\rangle \label{eq:integral_theta4}\\
&=2\pi\sum_{k=0}^m \frac{(-r^2)^k}{k !}\binom{m}{k},\label{eq:integral_theta5} 
\end{align}
where, in \eqref{eq:integral_theta2} we expand the exponentials, and in \eqref{eq:integral_theta4} the orthogonality of complex exponentials yields $2 \pi \delta_{{k,k^\prime}}$ and allows us to eliminate summation over $k^\prime$. We also apply the $(\hat{c}^\dagger)^k$ and $\hat{c}^k$ operators to $\bra{m}$ and $\ket{m^\prime}$, respectively. Orthogonality of the Fock states yields  \eqref{eq:integral_theta5}, which defines the Laguerre polynomial and yields the proof.
\end{proof}

\begin{lemma} \label{lemma:integral-theta2}
\begin{align}
&\int_{\theta=0}^{2\pi}d\theta (\cos(\theta) - \im\sin(\theta)) \notag \\
& \phantom{=}\times \langle m| e^{r(\cos(\theta) + \im\sin(\theta)) \hat{c}^{\dagger}} e^{-r(\cos(\theta) - \im\sin(\theta)) \hat{c}}  | m^\prime \rangle \notag \\ 
&= -2\pi r \sqrt{m}\sum_{k^\prime=0}^{m+1} (r^2)^{k^\prime}\frac{(-1)^{k^\prime}}{k^\prime!}\binom{m}{k^\prime}\frac{1}{k^\prime+1}.
\end{align}
\end{lemma}
\begin{proof}

The proof follows that of Lemma~\ref{lemma:integral-theta}, except that the orthogonality of complex exponentials yields $2\pi\delta_{k+1,k^\prime}$ instead of $2\pi\delta_{k,k^\prime}$.
\end{proof}

\bibliographystyle{IEEEtran}
\bibliography{./papers.bib}

\begin{thebibliography}{10}
\providecommand{\url}[1]{#1}
\csname url@samestyle\endcsname
\providecommand{\newblock}{\relax}
\providecommand{\bibinfo}[2]{#2}
\providecommand{\BIBentrySTDinterwordspacing}{\spaceskip=0pt\relax}
\providecommand{\BIBentryALTinterwordstretchfactor}{4}
\providecommand{\BIBentryALTinterwordspacing}{\spaceskip=\fontdimen2\font plus
\BIBentryALTinterwordstretchfactor\fontdimen3\font minus \fontdimen4\font\relax}
\providecommand{\BIBforeignlanguage}[2]{{%
\expandafter\ifx\csname l@#1\endcsname\relax
\typeout{** WARNING: IEEEtran.bst: No hyphenation pattern has been}%
\typeout{** loaded for the language `#1'. Using the pattern for}%
\typeout{** the default language instead.}%
\else
\language=\csname l@#1\endcsname
\fi
#2}}
\providecommand{\BIBdecl}{\relax}
\BIBdecl

\bibitem{bash15covertcommmag}
B.~A. Bash, D.~Goeckel, S.~Guha, and D.~Towsley, ``Hiding information in noise: Fundamental limits of covert wireless communication,'' \emph{{IEEE} Commun. Mag.}, vol.~53, no.~12, 2015.

\bibitem{bash13squarerootjsac}
B.~A. Bash, D.~Goeckel, and D.~Towsley, ``Limits of reliable communication with low probability of detection on {AWGN} channels,'' \emph{{IEEE} J. Sel. Areas Commun.}, vol.~31, no.~9, pp. 1921--1930, Sep. 2013.

\bibitem{bloch15covert}
M.~R. Bloch, ``Covert communication over noisy channels: A resolvability perspective,'' \emph{IEEE Trans. Inf. Theory}, vol.~62, no.~5, pp. 2334--2354, May 2016.

\bibitem{wang15covert}
L.~Wang, G.~W. Wornell, and L.~Zheng, ``Fundamental limits of communication with low probability of detection,'' \emph{IEEE Trans. Inf. Theory}, vol.~62, no.~6, pp. 3493--3503, Jun. 2016.

\bibitem{bash15covertbosoniccomm}
B.~A. Bash, A.~H. Gheorghe, M.~Patel, J.~L. Habif, D.~Goeckel, D.~Towsley, and S.~Guha, ``Quantum-secure covert communication on bosonic channels,'' \emph{Nat. Commun.}, vol.~6, Oct. 2015.

\bibitem{bullock20discretemod}
M.~S. Bullock, C.~N. Gagatsos, S.~Guha, and B.~A. Bash, ``Fundamental limits of quantum-secure covert communication over bosonic channels,'' \emph{{{IEEE} J. Sel. Areas Commun.}}, vol.~38, no.~3, pp. 471--482, Mar. 2020.

\bibitem{gagatsos20codingcovcomm}
C.~N. Gagatsos, M.~S. Bullock, and B.~A. Bash, ``Covert capacity of bosonic channels,'' \emph{{IEEE} J. Sel. Areas Inf. Theory}, vol.~1, pp. 555--567, 2020.

\bibitem{azadeh16quantumcovert-isitarxiv}
A.~Sheikholeslami, B.~A. Bash, D.~Towsley, D.~Goeckel, and S.~Guha, ``Covert communication over classical-quantum channels,'' in \emph{Proc. {IEEE} Int. Symp. Inform. Theory ({ISIT})}, Barcelona, Spain, Jul. 2016.

\bibitem{bullockCovertCommunicationClassicalQuantum2023}
M.~S. Bullock, A.~Sheikholeslami, M.~Tahmasbi, R.~C. Macdonald, S.~Guha, and B.~A. Bash, ``Covert {{Communication}} over {{Classical-Quantum Channels}},'' arXiv:1601.06826v7 [quant-ph], Jul. 2023.

\bibitem{thomasEfficientGenerationEntangled2022}
P.~Thomas, L.~Ruscio, O.~Morin, and G.~Rempe, ``Efficient generation of entangled multiphoton graph states from a single atom,'' \emph{Nature}, vol. 608, no. 7924, pp. 677--681, Aug. 2022.

\bibitem{hensenLoopholefreeBellInequality2015}
B.~Hensen, H.~Bernien, A.~E. Dr{\'e}au, A.~Reiserer, N.~Kalb, M.~S. Blok, J.~Ruitenberg, R.~F.~L. Vermeulen, R.~N. Schouten, C.~Abell{\'a}n, W.~Amaya, V.~Pruneri, M.~W. Mitchell, M.~Markham, D.~J. Twitchen, D.~Elkouss, S.~Wehner, T.~H. Taminiau, and R.~Hanson, ``Loophole-free {{Bell}} inequality violation using electron spins separated by 1.3 kilometres,'' \emph{Nature}, vol. 526, no. 7575, pp. 682--686, Oct. 2015.

\bibitem{krutyanskiyEntanglementTrappedIonQubits2023}
V.~Krutyanskiy, M.~Galli, V.~Krcmarsky, S.~Baier, D.~A. Fioretto, Y.~Pu, A.~Mazloom, P.~Sekatski, M.~Canteri, M.~Teller, J.~Schupp, J.~Bate, M.~Meraner, N.~Sangouard, B.~P. Lanyon, and T.~E. Northup, ``Entanglement of trapped-ion qubits separated by 230 meters,'' \emph{Phys. Rev. Lett.}, vol. 130, p. 050803, Feb 2023.

\bibitem{takedaDeterministicQuantumTeleportation2013}
S.~Takeda, T.~Mizuta, M.~Fuwa, P.~{van Loock}, and A.~Furusawa, ``Deterministic quantum teleportation of photonic quantum bits by a hybrid technique,'' \emph{Nature}, vol. 500, no. 7462, pp. 315--318, Aug. 2013.

\bibitem{guhaRateLoss2015}
S.~Guha, H.~Krovi, C.~A. Fuchs, Z.~Dutton, J.~A. Slater, C.~Simon, and W.~Tittel, ``Rate-loss analysis of an efficient quantum repeater architecture,'' \emph{Phys. Rev. A}, vol.~92, p. 022357, Aug 2015.

\bibitem{dhara2023entangling}
P.~Dhara, D.~Englund, and S.~Guha, ``Entangling quantum memories via heralded photonic bell measurement,'' \emph{Phys. Rev. Res.}, vol.~5, p. 033149, Sep 2023.

\bibitem{azuma2023repeatersurvey}
K.~Azuma, S.~E. Economou, D.~Elkouss, P.~Hilaire, L.~Jiang, H.-K. Lo, and I.~Tzitrin, ``Quantum repeaters: From quantum networks to the quantum internet,'' \emph{Rev. Mod. Phys.}, vol.~95, p. 045006, Dec 2023.

\bibitem{scarani09rmpQKD}
V.~Scarani, H.~Bechmann-Pasquinucci, N.~J. Cerf, M.~Du\ifmmode~\check{s}\else \v{s}\fi{}ek, N.~L\"utkenhaus, and M.~Peev, ``The security of practical quantum key distribution,'' \emph{Rev. Mod. Phys.}, vol.~81, pp. 1301--1350, Sep. 2009.

\bibitem{Honjo08qkd}
T.~Honjo, S.~W. Nam, H.~Takesue, Q.~Zhang, H.~Kamada, Y.~Nishida, O.~Tadanaga, M.~Asobe, B.~Baek, R.~Hadfield, S.~Miki, M.~Fujiwara, M.~Sasaki, Z.~Wang, K.~Inoue, and Y.~Yamamoto, ``Long-distance entanglement-based quantum key distribution over optical fiber,'' \emph{Opt. Express}, vol.~16, no.~23, pp. 19\,118--19\,126, Nov 2008.

\bibitem{arrazolaCovertQuantumCommunication2016}
J.~M. Arrazola and V.~Scarani, ``Covert {{Quantum Communication}},'' \emph{Phys. Rev. Lett.}, vol. 117, no.~25, p. 250503, Dec. 2016.

\bibitem{tahmasbi19covertqkd}
M.~Tahmasbi and M.~R. Bloch, ``Framework for covert and secret key expansion over classical-quantum channels,'' \emph{Phys. Rev. A}, vol.~99, p. 052329, May 2019.

\bibitem{tahmasbi20bosoniccovertqkd-jsait}
------, ``Toward undetectable quantum key distribution over bosonic channels,'' \emph{IEEE J. Sel. Areas Inf. Theory}, vol.~1, no.~2, pp. 585--598, 2020.

\bibitem{tahmasbi20covertqkd}
------, ``Covert and secret key expansion over quantum channels under collective attacks,'' \emph{IEEE Trans. Inf. Theory}, vol.~66, no.~11, pp. 7113--7131, Nov. 2020.

\bibitem{sharmaBoundingEnergyconstrainedQuantum2018}
K.~Sharma, M.~M. Wilde, S.~Adhikari, and M.~Takeoka, ``Bounding the energy-constrained quantum and private capacities of phase-insensitive bosonic {{Gaussian}} channels,'' \emph{New J. Phys.}, vol.~20, no.~6, p. 063025, Jun. 2018.

\bibitem{knillSchemeEfficientQuantum2001}
E.~Knill, R.~Laflamme, and G.~J. Milburn, ``A scheme for efficient quantum computation with linear optics,'' \emph{Nature}, vol. 409, no. 6816, pp. 46--52, Jan. 2001.

\bibitem{wilde16quantumit2ed}
M.~Wilde, \emph{Quantum Information Theory}, 2nd~ed.\hskip 1em plus 0.5em minus 0.4em\relax Cambridge University Press, 2016.

\bibitem{weedbrook12gaussianQIrmp}
C.~Weedbrook, S.~Pirandola, R.~Garc\'{\i}a-Patr\'on, N.~J. Cerf, T.~C. Ralph, J.~H. Shapiro, and S.~Lloyd, ``Gaussian quantum information,'' \emph{Rev. Mod. Phys.}, vol.~84, pp. 621--669, May 2012.

\bibitem{holevoEntanglementbreakingChannelsInfinite2008}
A.~S. Holevo, ``Entanglement-breaking channels in infinite dimensions,'' \emph{Probl. Inf. Transm.}, vol.~44, no.~3, pp. 171--184, Sep. 2008.

\bibitem{garciapatron2012majorization}
R.~Garc\'{\i}a-Patr\'on, C.~Navarrete-Benlloch, S.~Lloyd, J.~H. Shapiro, and N.~J. Cerf, ``Majorization theory approach to the gaussian channel minimum entropy conjecture,'' \emph{Phys. Rev. Lett.}, vol. 108, p. 110505, Mar 2012.

\bibitem{caruso2006weakdegradability}
F.~Caruso, V.~Giovannetti, and A.~S. Holevo, ``One-mode bosonic gaussian channels: a full weak-degradability classification,'' \emph{New J. Phys.}, vol.~8, no.~12, p. 310, dec 2006.

\bibitem{rosatiNarrowBoundsQuantum2018}
M.~Rosati, A.~Mari, and V.~Giovannetti, ``Narrow {{Bounds}} for the {{Quantum Capacity}} of {{Thermal Attenuators}},'' \emph{Nature Communications}, vol.~9, no.~1, p. 4339, Oct. 2018.

\bibitem{temme2010chi2}
K.~Temme, M.~J. Kastoryano, M.~Ruskai, M.~M. Wolf, and F.~Verstraete, ``The $\chi^2$-divergence and mixing times of quantum {M}arkov processes,'' \emph{J. Math. Phys.}, vol.~51, no.~12, p. 122201, 2010.

\bibitem{orszag16quantumotpics}
M.~Orszag, \emph{Quantum Optics}, 3rd~ed.\hskip 1em plus 0.5em minus 0.4em\relax Berlin, Germany: Springer, 2016.

\bibitem{kish2023comparison}
S.~P. Kish, P.~Gleeson, P.~K. Lam, and S.~M. Assad, ``Comparison of discrete variable and continuous variable quantum key distribution protocols with phase noise in the thermal-loss channel,'' arXiv:2206.13724 [quant-ph], 2023.

\bibitem{bennet96qec}
C.~H. Bennett, D.~P. DiVincenzo, J.~A. Smolin, and W.~K. Wootters, ``Mixed-state entanglement and quantum error correction,'' \emph{Phys. Rev. A}, vol.~54, pp. 3824--3851, Nov 1996.

\bibitem{grasslCodesQuantumErasure1997}
M.~Grassl, T.~Beth, and T.~Pellizzari, ``Codes for the {{Quantum Erasure Channel}},'' \emph{Physical Review A}, vol.~56, no.~1, pp. 33--38, Jul. 1997.

\bibitem{varnavaLossToleranceOneWay2006}
M.~Varnava, D.~E. Browne, and T.~Rudolph, ``Loss {{Tolerance}} in {{One-Way Quantum Computation}} via {{Counterfactual Error Correction}},'' \emph{Physical Review Letters}, vol.~97, no.~12, p. 120501, Sep. 2006.

\bibitem{caruso2006degradability}
F.~Caruso and V.~Giovannetti, ``Degradability of bosonic gaussian channels,'' \emph{Phys. Rev. A}, vol.~74, p. 062307, Dec 2006.

\bibitem{davis18constrainedcap}
N.~Davis, M.~E. Shirokov, and M.~M. Wilde, ``Energy-constrained two-way assisted private and quantum capacities of quantum channels,'' \emph{Phys. Rev. A}, vol.~97, p. 062310, Jun 2018.

\bibitem{goodenough2016quantumrepeat}
K.~Goodenough, D.~Elkouss, and S.~Wehner, ``Assessing the performance of quantum repeaters for all phase-insensitive gaussian bosonic channels,'' \emph{New Journal of Physics}, vol.~18, no.~6, p. 063005, jun 2016.

\bibitem{khatriInformationtheoreticAspectsGeneralized2020}
S.~Khatri, K.~Sharma, and M.~M. Wilde, ``Information-theoretic aspects of the generalized amplitude-damping channel,'' \emph{Phys. Rev. A}, vol. 102, p. 012401, Jul 2020.

\bibitem{berk06MODTRAN}
A.~Berk, G.~P. Anderson, P.~K. Acharya, L.~S. Bernstein, L.~Muratov, J.~Lee, M.~Fox, S.~M. Adler-Golden, J.~H. Chetwynd, Jr., M.~L. Hoke, R.~B. Lockwood, J.~A. Gardner, T.~W. Cooley, C.~C. Borel, P.~E. Lewis, and E.~P. Shettle, ``{MODTRAN5:} 2006 update,'' in \emph{Proc. {SPIE}}, vol. 6233, 2006, pp. 62\,331F--62\,331F--8.

\bibitem{wangCharacterizationCovertCapacity2022}
S.-Y. Wang, T.~Erdo{\u g}an, and M.~Bloch, ``Towards a {{Characterization}} of the {{Covert Capacity}} of {{Bosonic Channels}} under {{Trace Distance}},'' in \emph{Proc. IEEE Int. Symp. Inform. Theory (ISIT)}, Jun. 2022, pp. 318--323.

\bibitem{Chen}
K.~C. Chen, P.~Dhara, M.~Heuck, Y.~Lee, W.~Dai, S.~Guha, and D.~Englund, ``Zero-added-loss entangled-photon multiplexing for ground- and space-based quantum networks,'' \emph{Phys. Rev. Appl.}, vol.~19, p. 054029, May 2023.

\bibitem{noh2018capacitybounds}
K.~Noh, V.~V. Albert, and L.~Jiang, ``Quantum capacity bounds of gaussian thermal loss channels and achievable rates with {Gottesman-Kitaev-Preskill} codes,'' \emph{IEEE Trans. Inf. Theory}, vol.~65, no.~4, pp. 2563--2582, 2019.

\end{thebibliography}

\end{document}